\documentclass[reqno]{amsart}
\usepackage{amsfonts}
\usepackage{amsthm}
\usepackage{amssymb}
\usepackage{dsfont}
\usepackage{amscd}
\usepackage{color}

\newtheorem{theorem}{Theorem}[section]
\newtheorem{corollary}[theorem]{Corollary}

\newtheorem{prop}[theorem]{Proposition}
\theoremstyle{definition}

\newtheorem{example}[theorem]{Example}

\theoremstyle{remark}
 \newtheorem{remark}[theorem]{Remark}
\newtheorem{Conj}[theorem]{Conjecture}

\newtheorem*{acknow}{Acknowledgments}
\numberwithin{equation}{section}

\def\ra{\rightarrow}
\def\tr{\textrm{Tr}}

\def\be{\begin{equation}}
\def\ee{\end{equation}}
\def\bc{\begin{center}}
\def\ec{\end{center}}

\begin{document}

\title[Singular values for  products of two coupled random matrices]{Singular values for   products  of two coupled random matrices: hard edge phase transition}


\author{Dang-Zheng Liu} \address{Key Laboratory of Wu Wen-Tsun Mathematics, Chinese Academy of Sciences, School of Mathematical Sciences, University of Science and Technology of China, Hefei 230026, P.R.~China}
\email{dzliu@ustc.edu.cn}

\date{\today}



\begin{abstract} Consider  the product $GX$ of two  rectangular complex random matrices coupled by a constant matrix $\Omega$,  where $G$ can be thought to be  a Gaussian matrix and $X$  is a bi-invariant  polynomial ensemble. We  prove that the squared singular values form a biorthogonal  ensemble in Borodin's sense,  and further that  for $X$ being Gaussian the correlation kernel can be expressed as a double contour integral.  When all but finitely many eigenvalues of $\Omega^{} \Omega^{*}$ are   equal, the corresponding correlation
kernel is shown to admit a phase transition phenomenon at the hard edge in four different regimes as the coupling matrix changes. Specifically, the four limiting  kernels in turn are the Meijer G-kernel for  products of two independent  Gaussian matrices, a new critical and interpolating kernel, the perturbed Bessel kernel and the finite coupled product kernel associated with $GX$. In the special case that   $X$ is also a Gaussian matrix  and $\Omega$ is scalar, such a product  has been recently investigated by Akemann  and Strahov. We also propose  a Jacobi-type    product and prove the same  transition.
\end{abstract}

\maketitle
\section{Introduction and main results} \label{sectionintroduction}

\subsection{Introduction}

Given two complex matrices   $X_1$ of size $L\times M$  and  $X_2$ of size $M\times N$ with $L, M\geq N$, our interest in the present paper is the joint probability density function (PDF for short) which reads 
 \be \label{matrixpdf} P(X_1,X_2)=Z^{-1}  \exp\!\big\{ -\alpha \tr(X_{1}^{}X_{1}^{*}+X_{2}^{*} X_{2}^{} )+\tr(\Omega X_1 X_2+  (\Omega X_1 X_2)^*  )
  \big\}\ee
  with respect to  Lebesgue measure $dX_1 d X_2$ on $\mathbb{R}^{2(L+M)N}$. 
  Here $\alpha>0$ and  $\Omega$ is a non-random $N\times L$ matrix  as a coupling of  $X_1$  and  $X_2$ such that $\Omega \Omega^{*}<\alpha^2 I_N$, and the normalization
   \be Z=\Big({ \pi\over \alpha }\Big)^{(L+N)M} {\det}^{-M}{\Big(I_{L}-\frac{1}{\alpha^2}\Omega^{*}\Omega^{}\Big)},\ee
   where $I_{L}$ denotes an identity matrix of size $L\times L$. More precisely, our aim is to study the exact   functional form of the joint PDF and correlation kernel for   squared singular values of the matrix product $Y_2 = X_1 X_2$, and also to investigate  scaling limits at the hard edge. For other  local statistical properties such as  bulk and soft-edge limits, we  leave them  to a forthcoming paper.

When all involved matrices are real, the two-matrix  model defined in \eqref{matrixpdf} is very closely related to testing independence  and canonical correlation analysis
 in Multivariate Statistical Theory; see  for example the two excellent monographs   \cite[Chapters 9, 12 \& 13]{A03} and \cite[Chapter 11]{M82}, and \cite{JO15} for recent developments. To be exact, putting $X_1$ and $X_{2}$  together we have a sample covariance matrix \be X=\begin{pmatrix}  X_1\\
      X_{2}^{t}\end{pmatrix} \nonumber \ee  distributed according to
 \be \mathrm{const} \cdot \exp\bigg\{-\frac{1}{2} \tr
  \begin{pmatrix} X_{1}^{t} & X_{2}\end{pmatrix}
      {\begin{pmatrix}
          \Sigma_{11} & \Sigma_{12} \\
         \Sigma_{21} & \Sigma_{22}
        \end{pmatrix}}^{-1}
        \begin{pmatrix}  X_1\\
      X_{2}^{t}\end{pmatrix}  \bigg\}. \label{twomatrixreal}\ee
 Then the sample canonical correlations are defined to be the square roots of   eigenvalues of the \textit{ sample canonical correlation matrix }
 \be (X_{2}^{t}X_{2}^{})^{-1} (X_{1}X_{2})^{t} (X_{1}^{}X_{1}^{t})^{-1} X_{1}X_{2}, \label{canonicalcorrelation}\ee
 see  \cite[Sect. 13.4]{A03} or \cite[Sect. 11.3]{M82}.  When  $X_1$ and $X_2$ are independent and also both have identity covariance matrices, the canonical correlation matrix  \eqref{canonicalcorrelation} is just the so-called Jacobi/MANOVA  ensemble (see e.g. \cite{BHPZ14} or \cite[Chapt. 3.6]{Fo10}) and has been extensively studied, see \cite{Joh08,Wa80} and references therein.  For \eqref{canonicalcorrelation} associated with the general PDF \eqref{twomatrixreal} but restricted to a small rank of the population cross-covariance matrix $\Sigma_{12}$,  we refer the reader to \cite{BHPZ14} and \cite{JO15} for relevant investigations. Turning to the complex  counterpart of  \eqref{twomatrixreal}, the joint PDF \eqref{matrixpdf} can be re-expressed as
 \be \mathrm{const} \cdot \exp\Big\{-  \tr
  \begin{pmatrix} X_{1}^{*} & X_{2}   \nonumber \end{pmatrix}
      \Sigma^{-1}
        \begin{pmatrix}  X_1\\
      X_{2}^{*}\end{pmatrix}  \Big\}  \label{twomatrixcomplex}\ee
 where \be \Sigma=\begin{pmatrix}
           (\alpha I_{L}-\alpha^{-1}\Omega^{*}\Omega)^{-1}  & \alpha^{-1} \Omega^{*} (\alpha I_{N}-\alpha^{-1}\Omega\Omega^{*})^{-1} \\
      \alpha^{-1} \Omega (\alpha I_{L}-\alpha^{-1}\Omega^{*}\Omega)^{-1} & (\alpha I_{N}-\alpha^{-1}\Omega\Omega^{*})^{-1}
        \end{pmatrix}. \nonumber \ee
 In this case we will come back  to study the  complex analog of the  sample canonical correlations \eqref{canonicalcorrelation} and some possible relationships with the complex \textit{sample cross-covariance matrix}   $X_1X_2$  in the future. Finally, we  remark that the global spectral density  of   $X_1X_2$  under the joint PDF \eqref{twomatrixreal}
  with $L=N$   has been   investigated in  \cite{VB14} (this was pointed out to us by Gernot Akemann and Mario Kieburg).

However, our major motivation to consider \eqref{matrixpdf} comes from the paper by  Akemann  and Strahov \cite{AS15}  where $L=N$ and $\Omega$ is  set to be  a scalar matrix.  In this special case,  it was applied to Quantum Chromodynamics (QCD) with a baryon chemical potential by Osborn \cite{O04} where the complex
eigenvalues  were  determined and a limiting interpolation kernel between the Bessel kernel and the corresponding kernel of complex eigenvalues was derived. This important example   inspired  Akemann  and Strahov to turn to study  the singular values for products of two coupled random matrices.   Below we just give  a brief description of the Osborn-Akemann-Strahov model and refer the reader to  \cite{AS15} for more details.
 Let $A$ and $B$ be two independent $N\times M$ matrices with i.i.d. standard complex Gaussian entries, Osborn \cite{O04}  investigated  an analogue of the Dirac operator    in the context of QCD with   a baryon  chemical potential and  introduced a random matrix ensemble   with a coupling parameter $\mu\in [0,1]$
\be D=        \begin{pmatrix}
          0 & iA+\mu B \\
          iA^{*}+\mu B^{*} & 0 \\
        \end{pmatrix}. \nonumber
\ee
 He further calculated  complex eigenvalues of $D$ by reducing them to those of the product $(iA+\mu B)(iA^{*}+\mu B^{*})$. Equivalently, when turning to use the notation in \cite{AS15}, set
\be X_{1}=\frac{1}{\sqrt{2}}(A-i\sqrt{\mu}B), \qquad X_{2}=\frac{1}{\sqrt{2}}(A^{*}-i\sqrt{\mu}B^{*}), \nonumber \ee
 then $X_1$ and $X_2$ have a joint PDF as defined in \eqref{matrixpdf} but with $L=N$,  $\alpha=(1+\mu)/(2\mu)$ and $\Omega=(1-\mu)/(2\mu)I_{N}$; see \cite[Sect. 2]{AS15}. So far, as for    singular values of  the product matrix  $Y_2=X_1 X_2$,  a very interesting observation from  Akemann  and Strahov is that as $\mu \to 0$ it is equivalent to the classical Laguerre Unitary Ensemble (also called complex sample covariance matrices) while as $\mu \to 1$ it corresponds to  the product of two independent Gaussian random matrices; see \cite[Sect. 3]{AS15} for detailed discussion.

Another motivation why to study  products of coupled random matrices is that they are natural generalizations of products of independent random  matrices, as interpreted in \cite{AS15}.
   Actually, the topic on products of independent random matrices has attracted tremendous interest in recent years, largely  because of the finding of  exact solvability for Gaussian matrices \cite{AIK13,AKW13} and the appearing of some new families of universal patterns \cite{Fo14,FL15b,KKS15,KS14,KZ}. These also afford more examples to support the Wigner-Dyson \textit{Universality Conjecture}; see \cite{LWZ14} for the local statistical properties in the  bulk and at the soft edge. For a recent survey, see  \cite{AI15} and references therein. Interestingly, entirely different from the extensively  studied products,  the  singular values for  products of two coupled random  matrices no longer form a polynomial ensemble (that is, at least one of the two determinants consisting of the joint PDF is the Vandermonde determinant, cf. \cite{KS14}),  but a biorthogonal ensemble  with   both two sets of ``nontrivial" functions \cite{Bo98}; see \cite[Sect.3]{AS15} or   Proposition \ref{twoGpdf} below. In this sense, the result derived by Akemann and Strahov affords a very nice example of  biorthogonal ensembles, see    Borodin and P\'{e}ch\'{e}'s paper \cite{BP08} for another example of the generalized Wishart ensemble distributed proportionally as $\exp\{ -\tr(S_{1}X^{}X^{*}+S_{2}X^{*} X^{} )\}$, where $S_1, S_2$ are non-random $N\times N$ positive definite matrices while $X$ is random with the same size.

   Now let's return to the initial object \eqref{matrixpdf}.   More generally, we  can turn to consider the product of two coupled random matrices with matrix entries distributed proportionally as
 \be  \label{twomatrixpdfV}
   \exp\big\{ -\alpha \tr(G^{}G^{*})+\tr\big(\Omega GX+  (\Omega GX)^*\big)-\mathrm{Tr}V(X^{*}X^{})\big\} dG dX, \ee
where $dG=\prod_{j=1}^{L} \prod_{k=1}^{M} d\textrm{Re}\,{G_{j,k}} d\textrm{Im}\, G_{j,k}$,  $dX=\prod_{j=1}^{M} \prod_{k=1}^{N} d\textrm{Re}\,{X_{j,k}} d\textrm{Im}\, X_{j,k}$, and $V$ is a polynomial with positive leading coefficient. We will show that the squared singular values of $GX$  have a bi-orthogonal structure; see Corollary \ref{Vpdf} in Sect. \ref{sectpdf} below.   When $L=N$ and $\Omega$ is a scalar matrix,   the joint PDF \eqref{twomatrixpdfV} is usually called a coupled chiral two-matrix model and was first introduced   by Akemann, Damgaard, Osborn and Splittorff \cite{ADOS07} as a chiral analogue of Eynard-Mehta coupled Hermitian  matrix model \cite{EM98}. In this case, Akemann et al. derived the  joint PDF of squared singular values of $G$ and $X$ and also explicit formulas for all  spectral correlation functions, which opens up the possibility of   asymptotic analysis for local statistics; see e.g. \cite{ADOS07} and \cite{DGZ13}. However,  when $L>N$ (at this stage $\Omega$ must be a rectangular matrix due to the existence of the trace operation in the exponent) or $\Omega$ is not scalar, to the best of our knowledge, there are no explicit formulas  available for the  joint PDF of squared singular values of $G$ and $X$. But, once we focus on the product $GX$, its  singular values  can be exactly expressed as determinantal point processes.
At this time,   since $G$ is a Gaussian random matrix given that  $X$ is fixed, the joint PDF \eqref{twomatrixpdfV} can be treated as a \textit{coupled  multiplication} with a Ginibre matrix. Here it is worth emphasizing that the coupled case  only  preserves  biorthogonal ensembles of squared singular values, but not polynomial ensembles; see Theorem \ref{coupledmultiplication} in Sect. \ref{sectpdf} below.  This is different from the multiplication with a Ginibre matrix which transforms one polynomial ensemble  to another; see \cite{KS14} or \cite{CKW15,KKS15,Ku15}  for some   nice transformation identities  of polynomial ensembles.

The remainder of this article is organized as follows.  In the following subsection we summarise the main results on the joint eigenvalue PDF, correlation kernel and scaled kernel  for the product of two coupled Gaussian matrices defined in  \eqref{matrixpdf}.  In particular, there exists a hard-edge  transition phenomenon in four different regimes.
 Sect. \ref{sectpdf} is devoted to the joint PDF of squared singular values for coupled products of  a Ginibre (or Jacobi-type)  matrix    and a bi-unitarily invariant random  matrix, which includes \eqref{matrixpdf} and \eqref{twomatrixpdfV} as special cases.   The proofs of Theorems \ref{kernelpdf} and \ref{hardlimits} in Sect. \ref{subsectionmain} below  are respectively  given in Sect. \ref{sectkernel} and Sect. \ref{sectionhardlimit}, where the corresponding results   are also obtained for a Jacobi-type product.   In Sect.  \ref{4kernels}  further discussions  on the four limiting kernels are presented.

\subsection{Main results} \label{subsectionmain} Let  $\nu=M-N \geq 0, \kappa=L-N \geq 0$,    and let $\delta_1, \ldots, \delta_N$ be singular values of $\Omega$ such that $0\leq \delta_j<\alpha$ for $j=1,\ldots,N$.  Also let   $\Delta(x)=\prod_{1\leq i<j\leq N}(x_j-x_i)$ denote the Vandermonde determinant.  We will frequently use
 two kinds of modified  Bessel functions  defined by
\be I_{\nu}(z)=  \sum_{k=0}^{\infty} \frac{1}{ k!\Gamma(\nu+1+k)} \Big(\frac{z}{2}\Big)^{2k+\nu}   \label{Ifunction}\ee
and \be K_{\nu}(z)=  \frac{1}{2}\Big(\frac{z}{2}\Big)^{\nu} \int_{0}^{\infty} t^{-\nu-1} e^{-t-\frac{z^2}{4t}} dt,  \quad |\textrm{arg}(z)|<\frac{\pi}{4},  \label{Kfunction} \ee
(cf. \cite[8.432.6]{GR07}) and also the hypergeometric function $_0F_1$   defined by
\be _0F_1(\nu+1;z)=\sum_{k=0}^{\infty} \frac{1}{ (\nu+1)_k} \frac{z^k}{k!}, \label{0F1I}\ee
where the Pochhammer symbol $(a)_{k}=a(a+1)  \cdots (a+k-1)$.  Viewed from  the integral representation \eqref{Kfunction},   the argument of  $K_{\nu}(z)$  is restricted to the interval $(-\pi/4, \pi/4)$, however,   it  can be analytically extended to the domain $\mathbb{C}\backslash (-\infty,0]$, see e.g. \cite[10.25]{Olver10}.  It is worth noting the two relations
\be _0F_1(\nu+1;z)=\Gamma(\nu+1)(\sqrt{z})^{-\nu} I_{\nu}(2\sqrt{z}) \label{0F1I}\ee
and \be K_{-\nu}(z)=K_{\nu}(z),\ee
 which respectively show that the RHS of \eqref{0F1I} is an entire function of $z$ and $K_{\nu}(z)$ has even parity in its parameter. 
Here and below the principal square root of a nonzero complex number $z$ is denoted by $\sqrt{z}$ as in the positive real case.

Our first result is an exact formula for the joint  PDF of squared singular values of  $X_{1}X_{2}$ under   \eqref{matrixpdf}, in which    modified Bessel functions  $I_{\nu}$ and  $K_{\nu}$ are involved; see Theorem \ref{coupledmultiplication} of Sect. \ref{sectpdf}   for the more general results.
 \begin{prop} \label{twoGpdf} With the joint PDF of two matrices $X_1$ and $X_2$ defined in  \eqref{matrixpdf}, let  $Y_2=X_1 X_2$. Then the joint PDF for    the squared singular values  of $Y_{2}$  on $[0,\infty)^{N}$ is given by
 \be \label{eigenvaluepdfGinibre} \mathcal{P}_{N}(x_1,\ldots,x_N)=\frac{1}{Z_N} \det\Big[I_{\kappa}(2\delta_i  \sqrt{x_j})\Big]_{i,j=1}^{N} \det\Big[x^{\frac{\nu+i-1}{2}}_{j} K_{\nu-\kappa+i-1}(2\alpha  \sqrt{x_j})\Big]_{i,j=1}^{N},
 \ee
 where $0\leq \delta_{j}<\alpha$ for $j=1, \ldots, N$ and  the normalization  constant  
  \be Z_N= N! 2^{-N}\alpha^{-N(\kappa+\nu+N)-\frac{1}{2}N(N+1)} \Delta(\delta^2) \prod_{j=1}^{N}\Big(\Gamma(j+\nu) \delta_{j}^{\kappa} \big(1-  \frac{\delta_{j}^{2}}{\alpha^2} \big)^{-\nu-N}\Big).\ee
  \end{prop}

We stress that when some of the $\delta_j$'s  coincide,  L'H\^{o}spital's rule provides the appropriate eigenvalue density. When $\kappa=0$ and all $\delta_j$'s are equal, the PDF in Proposition \ref{twoGpdf}   has been derived by Akemann and Strahov; see \cite[Theorem 3.1]{AS15}.

For  the joint eigenvalue PDF \eqref{eigenvaluepdfGinibre} above as a determinantal  point process,    we  find a double contour integral expression for the correlation kernel, which provides the starting point for further asymptotic analysis. Both our double integral formula and its derivation are  different from those  given by  Akemann and  Strahov, see   \cite{AS15} for exact formulae and brilliant  derivations. Therein, the authors discussed in details  biorthogonal functions, five-term recurrence relations, Christoffel-Darboux formula and relevant contour integral representations.
 \begin{theorem}\label{kernelpdf} 
% %With the joint  PDF defined in  \eqref{matrixpdf}, let  $Y_2=X_1 X_2$. Then the joint eigenvalue  PDF    for $Y_{2}^{*} Y_{2}^{}$ can be  written in the form
%% \be \mathcal{P}_{N}(x_1,\ldots,x_N)=\frac{1}{N!} \det[K_N(x_i,x_j)]_{i,j=1}^{N},
%% \ee
The correlation kernel  for the biorthogonal ensemble  \eqref{eigenvaluepdfGinibre} is given by 
 \begin{multline}
K_N(x,y)=\frac{2\alpha^2}{(2\pi i)^2}\int_{\mathcal{C}_{\mathrm{out}}} du \int_{\mathcal{C}_{\mathrm{in}}} dv\,K_{-\kappa}(2\alpha\sqrt{(1-u)x})\,  I_{\kappa}(2\alpha\sqrt{(1-v)y}) \\  \times \frac{1}{u-v}\Big(\frac{1-u}{1-v}\Big)^{\kappa/2}\Big(\frac{u}{v}\Big)^{-\nu-N}\prod_{l=1}^{N}\frac{u-(1-\delta_{l}^{2}/\alpha^2)}{v-(1-\delta_{l}^{2}/\alpha^2)}, \label{kernelCD}\end{multline}
where  $\mathcal{C}_{\mathrm{in}}$ is a counterclockwise contour encircling $1-\delta_{1}^{2}/\alpha^2,\ldots, 1-\delta_{N}^{2}/\alpha^2$, and $\mathcal{C}_{\mathrm{out}}$   is a   simple  contour  counterclockwise around the origin with  $\mathrm{Re}(z)<1$ for $z\in \mathcal{C}_{\mathrm{out}}$
   such that       
   $\mathcal{C}_{\mathrm{in}}$ is entirely  to the right side of  $\mathcal{C}_{\mathrm{out}}$.  When $0< \delta_{j}<\alpha$ for $j=1, \ldots, N$, we can also choose contours such that  
   $\mathcal{C}_{\mathrm{in}}$ is  contained entirely in  $\mathcal{C}_{\mathrm{out}}$.
    \end{theorem}

Note that    it is unnecessary to assume  $\mathrm{Re}(z)<1$ for $z\in \mathcal{C}_{\mathrm{in}}$  in   \eqref{kernelCD},  unlike  $\mathcal{C}_{\mathrm{out}}$, since  $z^{-\kappa/2} I_{\kappa}(2\sqrt{z})$ is an entire function of $z$ (cf.  eqn \eqref{0F1I}).  Besides, we will select  more specific  contours  as required  in investigating  the scaling limits of correlation kernels.

Next, we  focus on asymptotic behavior of the correlation kernel under the assumptions of finite-rank  perturbation of the matrix $\Omega$ and $\mu$-dependent coupling (see \cite{AS15} for discussion in details), i.e., for a given nonnegative integer $m$ independent of $N$,
\be \delta_{m+1}=\cdots=\delta_N=\delta\  \mbox{and}\   \alpha=(1+\mu)/(2\mu),  \  \delta=(1-\mu)/(2\mu), \ 0<\mu\leq 1, \label{finiterank}\ee
where $\mu=\mu_N$ may depend on $N$ but will be used without the subscript for simplicity, unless otherwise specified.
Our  main results  are devoted to    hard edge scaling limits at different  scales of $\mu$, and particularly to  a critical kernel after a double scaling.

For nonnegative integers $\nu, \kappa$ and $m$, we  introduce four types of double integrals for correlation kernels  as follows.
The first kernel  is defined to be
 \begin{align}
 K_{\mathrm{I}}(\xi,\eta) &= \left(  \frac{\eta}{\xi}\right)^{\kappa/2}  \frac{1}{2\pi i}\int_{0}^{\infty}  dt  \int_{\mathcal{C}_{\mathrm{0}}}  ds  \,     t^{\kappa-1} s^{-\kappa-1}
 e^{s-t}  \nonumber \\ & \quad \times    4 \left(  \frac{\xi s}{\eta t}\right)^{\nu/2} K_{\mathrm{\nu}}^{(\mathrm{Bes})}\Big(\frac{4\eta}{s}, \frac{4\xi}{t}\Big), \label{kernelsub}\end{align}
 where $\mathcal{C}_{\mathrm{0}}$ is a   counterclockwise contour around the origin and  the  Bessel kernel
 \be K_{\mathrm{\nu}}^{(\mathrm{Bes})}\big(x,y\big)
 =\frac{J_{\nu}(\sqrt{x})\sqrt{y}J'_{\nu}(\sqrt{y})-J_{\nu}(\sqrt{y})\sqrt{x}J'_{\nu}(\sqrt{x})}{2(x-y)}\ee
with  the Bessel function of the first kind $J_{\nu}$; cf. \cite{Fo93,TW94}.   Note that  this  type of convolution representation in  \eqref{kernelsub} has been obtained in  the product of two independent random matrices   for finite   matrix size $N$,  see \cite[Theorem 2.8(b)]{CKW15}.    Actually, in  Sect. \ref{4kernels} below  this will prove  to be      the Meijer G-kernel associated with the  product of two independent  Gaussian matrices which appeared previously in \cite{BGS14,KZ}. 

The second one is a new critical and interpolating kernel between the Meijer G-kernel and the perturbed Bessel kernel, which reads  for $\tau>0$ and $\pi_{1},\ldots,  \pi_m \in (0, 1)$
\begin{align}
K_{\mathrm{II}}(\tau;\xi,\eta)&= \frac{2}{(2\pi i)^2}\int_{\mathcal{C}_{\mathrm{out}}} du \int_{\mathcal{C}_{\mathrm{in}}} dv \,  K_{-\kappa}(2\sqrt{(1-u)\xi})\,  I_{\kappa}(2\sqrt{(1-v)\eta}) \nonumber\\  &  \quad \times e^{-\frac{\tau}{u}+\frac{\tau}{v}}\frac{1}{u-v} \Big(\frac{1-u}{1-v}\Big)^{\kappa/2}\Big(\frac{u}{v}\Big)^{-\nu-m} \prod_{l=1}^{m}\frac{u- \pi_{l}}{v-\pi_{l}}. \label{kernelcrit}\end{align}
The last two kernels are the perturbed Bessel kernel  which was first defined in \cite{DF06} for $\pi_{1},\ldots,  \pi_m \in (0, \infty)$
\begin{align}
K_{\mathrm{III}}(\xi,\eta) &= \frac{2}{ (2\pi i)^2} \frac{1}{4(\xi \eta)^{\frac{1}{4}}} \int_{\mathcal{C}_{\mathrm{out}}} du \int_{\mathcal{C}_{\mathrm{in}}} dv \, e^{\sqrt{\xi}u-\sqrt{\eta}v-\frac{1}{u}+\frac{1}{v}}   \nonumber \\
 & \quad \times \frac{1}{u-v}\Big(\frac{u}{v}\Big)^{-\nu-m}\prod_{l=1}^{m}\frac{u- \pi_{l}}{v-\pi_{l}}, \label{kernelsup}\end{align}
 and the finite coupled product kernel with  $\pi_{1},\ldots,  \pi_m \in (0, 1)$ and  $m\geq 1$
\begin{align}
K_{\mathrm{IV}}(\xi,\eta)&= \frac{2}{(2\pi i)^2}\int_{\mathcal{C}_{\mathrm{out}}} du \int_{\mathcal{C}_{\mathrm{in}}} dv \,  K_{-\kappa}(2\sqrt{(1-u)\xi})\,  I_{\kappa}(2\sqrt{(1-v)\eta}) \nonumber\\  &\quad\times \frac{1}{u-v} \Big(\frac{1-u}{1-v}\Big)^{\kappa/2} \Big(\frac{u}{v}\Big)^{-\nu-m}\prod_{l=1}^{m}\frac{u- \pi_{l}}{v-\pi_{l}}.  \label{kernelsupsup}\end{align}
In the definition of last three kernels,  
  $\mathcal{C}_{\mathrm{out}}$   is a  simple  counterclockwise  contour around  the origin (with  $\mathrm{Re}(z)<1,  \forall z\in \mathcal{C}_{\mathrm{out}}$ for $K_{\mathrm{II}}$ and $K_{\mathrm{IV}}$)
and entirely within it  $\mathcal{C}_{\mathrm{in}}$ is a counterclockwise contour encircling $0, \pi_{1},\ldots,  \pi_m$. %In addition,  for $K_{\mathrm{II}}$ and $K_{\mathrm{IV}}$ we assume that  $\mathrm{Re}(z)<1$ for $z\in \mathcal{C}_{\mathrm{out}}$.  
Note that the last one is actually the correlation kernel  \eqref{kernelCD} associated with coupled products of two Gaussian matrices with properly chosen parameters; see Sect. \ref{4kernels} for detailed discussion on the four kernels.
Also, it's  worth emphasizing that   the kernels defined above may depend on parameters $\tau>0$, $\kappa$, $m$ and $\pi_1, \ldots, \pi_m$,  however, we still use the shorthand notations  for simplicity, unless specified.

 We are now ready to state the main results  which describe a transition of hard edge limits for correlation kernels in four different regimes, by tuning the scale  of $1-\delta_{1}^{2}/\alpha^2,\ldots, 1-\delta_{m}^{2}/\alpha^2$ as $\mu N$ varies from zero to infinity at different scales. A similar hard edge phase transition occurs in three different regimes for the shifted mean chiral Gaussian ensemble \cite{FL15b}. Recently, some different types of hard-to-soft edge transition have been observed for Gaussian perturbations of hard edge random matrix ensembles by Claeys and Doeraene \cite{CD16}.  Also, see \cite{BBP} for the famous   Baik-Ben Arous-P\'{e}ch\'{e} phase transition for largest eigenvalues.
 
\begin{theorem} [Hard edge limits] \label{hardlimits}  Assume that the parameters  $\delta_{j}$ satisfy the condition  \eqref{finiterank} and  $0\leq \delta_{j}<\alpha$ for $j=1, \ldots, m$. 
 With  the correlation kernel    \eqref{kernelCD}  and with fixed nonnegative integers $\nu$ and $\kappa$,  the following hold uniformly for any $\xi$ and $\eta$ in a compact set of $(0,\infty)$ as $N\rightarrow \infty$.

\begin{itemize}

 \item [(i)]  If $\mu N \ra \infty$, %  and    $1- \delta_{l}^{2}/\alpha^2$ for $l=1, \ldots, m,$
 then
 \be
\frac{\mu }{N} K_N\big( \frac{\mu }{N} \xi, \frac{\mu }{N} \eta\big)\ra K_{\mathrm{I}}(\xi,\eta). \ee

 \item [(ii)]  If $\mu N \ra \tau/4$ with $\tau>0$ and
$ 1- \delta_{l}^{2}/\alpha^2  \ra \pi_l \in (0, 1)$ for $l=1, \ldots, m,$ then
 \be
 \alpha^{-2} K_N( \alpha^{-2} \xi,  \alpha^{-2} \eta)\ra K_{\mathrm{II}}(\tau;\xi,\eta). \label{IIlimit} \ee

\item [(iii)]  If $\mu N \ra 0$ and  $ 1- \delta_{l}^{2}/\alpha^2= 4\mu N  \pi_l$  with $\pi_l\in (0, \infty)$ 
 for $l=1, \ldots, m,$ then
 \be
 %% \frac{e^{(\alpha/N)\sqrt{\xi}}}{e^{(\alpha/N)\sqrt{\eta}}}  e^{\alpha(\sqrt{\xi}-\sqrt{\eta})/N} e^{\frac{\alpha}{N}(\sqrt{\xi}-\sqrt{\eta})} 
   \frac{ e^{\frac{\alpha}{N} \sqrt{\xi}}}{e^{\frac{\alpha}{N} \sqrt{\eta}}}\frac{1}{4N^2}  K_N\big(\frac{1}{4N^2} \xi, \frac{1}{4N^2} \eta\big)\ra K_{\mathrm{III}}(\xi,\eta). \label{IIIlmit} \ee

\item [(iv)]  If $\mu N \ra 0$ and
$ 1- \delta_{l}^{2}/\alpha^2  \ra \pi_l\in (0,1)$ for $l=1, \ldots, m,$ then for $m\geq 1$
 \be
4\mu^2 K_N(4\mu^2 \xi, 4\mu^2 \eta)\ra K_{\mathrm{IV}}(\xi,\eta).  \label{IVlimit} \ee
\end{itemize}
 \end{theorem}

This theorem says that there are exactly four distinct limiting kernels  as the coupling strength $\mu$ changes, along with properly chosen scalings of parameters $\delta_1, \ldots, \delta_m$. The same result  also appears in a Jacobi-type product ensemble which predicts a universal pattern; see Theorem \ref{hardlimitsJ} in Sect. \ref{sectionhardlimit}.  Compared with all those known phase transition phenomena mentioned above in Random Matrix Theory (RMT), as far as we know, Theorem \ref{hardlimits} is the first show of a four-term transition.   
Usually in RMT the pattern of universality for  local  eigenvalue statistics  depends on some exponent $c<1$, with which the limiting density of eigenvalues diverges (hard edge) or vanishes (soft edge) like $|x-x_{0}|^{-c}$ as $x \ra x_{0}$  from either side. This  leads to a change in fluctuations in powers of matrix size $N$ and thus the scaling of the correlation kernel; see e.g. \cite{BBP,FL15b}. As to the kernels above, when $m=0$ it was argued in \cite[Sect. 2]{AS15b} that the exponent $c=2/3, 3/4, 3/4$ at the origin corresponding to cases (i), (ii) and (iii) of Theorem \ref{hardlimits} (if we change variables $\xi, \eta$  to $\xi^2, \eta^2$ in  cases (ii) and (iii), then $c=1/2$ in both cases, which  is consistent with the description given  in \cite{AS15b}).  Thus,   at least,  the limit from $K_{\mathrm{II}}$ to $K_{\mathrm{I}}$ is   a candidate for a phase transition.  On the other hand, the scaling $\alpha^{-2}\sim \tau^{2}(2N)^{-2}$ in case (ii)  is  the same  as in  case (ii) but different   from  cases (i) and (iv). This probably indicates  a phase transition  from $K_{\mathrm{I}}$ to $K_{\mathrm{II}}$ to $K_{\mathrm{IV}}$.

We remark that although  case (iv) can be formally obtained by merely permitting   $\tau=0$ in  case (ii),  we  sepatate  it   at least   for two reasons: 
one is, we divide  the limits  of  $\mu N$ into three  categories: $\infty$, $(0, \infty)$ and 0, the third of which is again   divided into two cases according to the choice of different   scalings of parameters $\delta_1, \ldots, \delta_m$; the other  is to  emphasize that  the finite coupled product kernel  $K_{\mathrm{IV}}$ will appear as a limiting kernel in RMT like the  finite GUE and LUE  kernels (cf. \cite{BBP,FL15b}), and that  it is non-trivial only  when  the finite rank perturbation $m\geq 1$.  A few other relevant remarks are as follows.
 \begin{remark} We noticed the preprint \cite{AS15b} when it appeared early during the drafting of this article. At that time Theorem \ref{kernelpdf} and  Parts (i), (ii) and (iv) of Theorem \ref{hardlimits} was completed while Part (iii) was later inspired by \cite[Theorem 1.5 (a)]{AS15b}. We are grateful to Gernot Akemann for detailed discussions on the main results of \cite{AS15b}.
 \end{remark}

 \begin{remark} When $L=N$ (that is, $\kappa=0$) and $\Omega$ is a scalar matrix (equivalently, $m=0$ in \eqref{finiterank}), we compare Theorem \ref{hardlimits} with relevant results of Akemann and Strahov as follows. For fixed $\mu$, Part (i) of Theorem \ref{hardlimits} was previously obtained by  Akemann and Strahov, see  \cite[Theorem 3.9]{AS15}.  In a subsequent paper  \cite{AS15b}, with  $\mu=gN^{-\chi}$,  they further obtained   the hard edge limits in   cases  $0\leq \chi<1$,  $\chi=1$  and $\chi>1$, which respectively corresponds to  Parts (i), (ii)  and (iii), and proved that the limiting kernels in  Parts (i) and (iii) agree with the standard integral forms.  Although their   double integral of correlation kernel  at the critical scale is different from  ours, these are believed to be the same; see Sect. \ref{4kernels} below for further discussion on the four kernels in Theorem \ref{hardlimits}.
%% (1) For fixed $\mu$, Part (i) of Theorem \ref{hardlimits} was previously obtained by  Akemann and Strahov, although the double integrals for limiting  kernels are in different form; see \cite[Theorem 3.9]{AS15}. (2) Let $\mu=gN^{-\chi}$.  In \cite{AS15b}, Akemann and Strahov  have proved   the case $0\leq \chi<1$ which is a special case of Part (i),  the critical case $\chi=1$  corresponding to Part (ii) and the case $\chi>1$ which is part of Part (iii).  Although their   double integral of correlation kernel  at the critical scale is different from  ours, these are believed to be the same; see Sect. \ref{4kernels} below for further discussion on the four kernels in Theorem \ref{hardlimits}.
 \end{remark}

\begin{remark} \label{normconvergence} Note that  for biorthogonal ensembles  the gap probability that no eigenvalues belong to a given Borel set $A\subset\mathbb{R}$ has a Fredholm determinant expression (see e.g. \cite[Lemma 3.2.4]{AGZ09})
   \be \mathbb{P}(x_1 \in A^c, \ldots, x_N \in A^c)=1+\sum_{k=1}^{\infty}\frac{(-1)^k}{k!}\int_{A} \cdots \int_{A}  \det[K_N(t_i,t_j)]_{i,j=1}^{k} dt_{1} \cdots dt_{k}, \nonumber \ee
   if we strengthen the results in Theorem   \ref{hardlimits}
  from uniform convergence into the trace norm convergence of the integral operators with respect to the correlation kernels,
    then as a direct consequence we  have   the limiting gap probabilities after rescaling,  especially including the   distribution of smallest  squared singular values (cf. \cite[Chapters 8 \& 9]{Fo10}). In the case of Part (iv), we have closed expression for  scaling limit of  the  smallest  squared singular values, see equation \eqref{goodform} in Sect. \ref{4kernels} below.   Since the proof of trace norm convergence is only a technical elaboration that confirms a well-expected result, we do not give the detail.
\end{remark}

Finally, we conclude this section with two conjectures. One is the product of two coupled  real Gaussian matrices while the other refers to generalizations of matrix entries from  Gaussian variables to the more  general random variables.

\begin{Conj} For real counterpart of the joint PDF \eqref{matrixpdf}, Theorem   \ref{hardlimits} still holds for different limiting kernels with certain   Pfaffian structure, but under the same scalings. In particular, the critical scale of $\mu$ is again expected to be $1/N$. \end{Conj}

To state the second conjecture, let $\alpha>0$ and   $m$ be  a fixed nonnegative integer, assume that   $\delta_1, \ldots, \delta_m$ and $\delta$ are complex numbers with absolute value less than $\alpha$. We consider two  complex random  matrices $X_1=[X_{1}(j,k)]_{1\leq j\leq N, 1\leq k\leq M}$ and $X_2=[X_{2}(k,j)]_{1\leq k\leq M,1\leq j\leq N}$ such that the    following conditions are satisfied:
\begin{enumerate}\item [(C1)]  The vector pairs  $\{X_{1}(j,k), X_{2}(k,j)\}_{1\leq j\leq N, 1\leq k\leq M}$  are independent,  and moreover   $\{X_{1}(j,k), X_{2}(k,j)\}_{m+1\leq j\leq N, 1\leq k\leq M}$  are   identically distributed and so are   $\{X_{1}(j,k), X_{2}(k,j)\}_{1\leq k\leq M}$   for any given $j\in \{1, \ldots, m\}$;
\item [(C2)] For any $j,k$,  $\mathds{E}[X_{1}(j,k)]=\mathds{E}[X_{2}(k,j)]=0$, $\mathds{E}[(X_{1}(j,k))^{2}]=\mathds{E}[(X_{2}(k,j))^{2}]=0$, $\mathds{E}[X_{1}(j,k) \overline{{X_{2}(k,j)}} ]=0$;
    \item [(C3)] When    $j\geq m+1$,  $\mathds{E}[|X_{1}(j,k)|^{2}]=\mathds{E}[|X_{2}(k,j)|^{2}]=\alpha/(\alpha^2-|\delta|^2)$ and  $\mathds{E}[X_{1}(j,k)  {X_{2}(k,j)}  ]= \bar{\delta}/(\alpha^2-|\delta|^2)$, while for any given $j\in \{1, \ldots, m\}$ $\mathds{E}[|X_{1}(j,k)|^{2}]=\mathds{E}[|X_{2}(k,j)|^{2}]=\alpha/(\alpha^2-|\delta_j|^2)$ and  $\mathds{E}[X_{1}(j,k)  {X_{2}(k,j)}  ]= \bar{\delta_{j}}/(\alpha^2-|\delta_j|^2)$;

\item [(C4)]   For any $j,k$,  $\mathds{E}[|X_{1}(j,k)|^{4}]<\infty$ and $\mathds{E}[|X_{2}(k,j)|^{4}]<\infty$.
\end{enumerate}
Note that the joint PDF \eqref{matrixpdf} with $L=N$  satisfies the above assumptions since $\Omega$ can  be taken to be diagonal according to the invariance of  Gaussian random variables.

\begin{Conj}  For $\alpha=(1+\mu)/(2\mu),   \delta=(1-\mu)/(2\mu), 0<\mu\leq 1$,  under the above assumptions (C1)-(C4),    Theorem   \ref{hardlimits} still holds true but with $\kappa=0$. \end{Conj}

\section{Coupled multiplication with a random matrix} \label{sectpdf}

\subsection{Coupled multiplication with a Ginibre matrix} \label{sectpdfG}
For complex matrices $G$ of size $L \times M$, $X$ of size $M \times N$ and $\Omega $ of size $N\times L$ with $L, M\geq N$, suppose that  the joint  probability distribution of $G$  and $X$   is equal to
 \be  \label{twomatrixpdf}
  Z^{-1}  \exp\big\{ -\alpha \tr(G^{}G^{*})+\tr(\Omega GX+  (\Omega GX)^*)\big\} \, h(X) dG dX, \ee
where $dG=\prod_{j=1}^{L} \prod_{k=1}^{M} d\textrm{Re}\,{G_{j,k}} d\textrm{Im}\, G_{j,k}$ and $dX=\prod_{j=1}^{M} \prod_{k=1}^{N} d\textrm{Re}\,{X_{j,k}} d\textrm{Im}\, X_{j,k}$, and also suppose that $h(X)$ is invariant under left and right multiplication with unitary matrices, i.e., $h(UXV)=h(X)$ for any unitary matrices  $U\in U(M)$ and $V\in U(N)$.  We turn to the product  $Y = GX$ and study the squared singular values of $Y$.

The main result of this section can be stated as follows.
\begin{theorem}\label{coupledmultiplication} With the joint PDF defined in  \eqref{twomatrixpdf}, let $\delta_1, \ldots, \delta_N$ be singular values of $\Omega$ such that $0\leq \delta_j<\alpha$ for $j=1,\ldots,N$. Suppose that  $f_{k}(t)$ ($k=1, \ldots, N$) are continuous in $(0,\infty)$ such that all  $e^{\alpha t}f_{k}(t)$  are bounded in $[0, \infty)$, let 
\be h(X)=\frac{1}{\Delta(t)}\det[f_{k}(t_j)]_{j,k=1}^N, \quad 0<t_1, \ldots, t_N<\infty, \label{hPE}\ee
 where  $t_1, \ldots, t_N$ are    eigenvalues of $X^{*} X^{}$,  then the  squared singular values    of $Y=GX$ have a joint PDF on $[0,\infty)^N$
 \be \label{eigenvaluepdf} \mathcal{P}_{N}(x_1,\ldots,x_N)=\frac{1}{Z_N} \det[\xi_i(x_j))]_{i,j=1}^{N} \det[\eta_i(x_j))]_{i,j=1}^{N},    %\det[x^{\frac{M-N}{2}}_{j} K_{i-1+\nu}(2\alpha  \sqrt{x_j})]_{i,j=1}^{N},
 \ee
 where  $\xi_{i}(z)= {_0F_1}(L-N+1;\delta_{i}^{2}z)$ and
  \be \eta_{i}(z)= \int_{0}^{\infty} t^{L-N} e^{-\alpha t} \big(\frac{z}{t}\big)^{M-N} f_{i}\big(\frac{z}{t}\big) \frac{dt}{t}. \ee
The normalization constant can be evaluated by
  \be Z_N=N! ((L-N)!)^N \alpha^{-N(L-N+1)}\det\Big[\int_{0}^{\infty} e^{x\delta_{j}^{2}/\alpha} x^{M-N}f_{k}(x) dx\Big]_{j,k=1}^N.\ee
  \end{theorem}

Theorem \ref{coupledmultiplication} shows that the coupled product of a complex Ginibre matrix and a  bi-invariant polynomial   ensemble produces a bi-orthogonal ensemble, with two sets of ``nontrivial" functions. This affords us random matrix realizations for a class of determinantal point processes which are bi-orthogonal ensembles but not polynomial ensembles.  The following proof is inspired  by these of \cite[Theorem 3.1]{AS15} and   \cite[Theorem 2.1]{KS14}.
\begin{proof} We proceed in three steps.

\textbf{Step 1: Reduction}.
We claim that the problem   can   be reduced to the study of case $M=N$. Let us assume that $M>N$. Then any matrix $X$ of size $M\times N$ can be decomposed as
\be X=U\binom{X_0}{O},
\ee
where $U$ is an $M\times M$ unitary matrix which can be uniquely taken to be in some specific form, $X_0$ is an $N\times N$ complex matrix and $O$ is a zero matrix of size $(M-N) \times N$; cf.  Lemma 2.1 and Appendix A in \cite{Fis12}. By  the results of \cite[Sect. 2]{Fis12}, we obtain the joint distribution of $G, X_0$ and $U$   proportional to   \be  \label{threematrixpdf}
  \det(X_{0}^{*}X_{0}^{})^{M-N}  \exp\Big\{ -\alpha \tr(G^{}G^{*})+\tr\big(\Omega GU\tbinom{X_0}{O}+  \big(\Omega G\tbinom{X_0}{O}\big)^*\big)\Big\} \, h(X_0) dG dX_0 [dU], \ee
where $[dU]$ denotes the induced  measure from the Haar measure of $M\times M$ unitary group; cf.  \cite[Eq.(6)]{Fis12}.

Make a change of variables $\widehat{G}=GU$ and rewrite $\widehat{G}=(G_0\ G_1)$ with two blocks $G_0$ of size $L\times N$ and   $G_1$ of size $L\times (M-N)$, then  $GX=G_0 X_0$ and  the joint distribution of $G_0, G_1, X_0$ and $U$ can be rewritten to be proportional   to
 \begin{align}  \label{fourmatrixpdf}
     \exp\big\{ -\alpha \tr(G_{0}^{}G_{0}^{*}+G_{1}^{}G_{1}^{*})+& \tr\big(\Omega G_0X_0+(\Omega G_0X_0)^*\big)\big\} \nonumber \\
     & \times \det(X_{0}^{*}X_{0}^{})^{M-N}\, h(X_0) dG_{0} dX_0 [dU] dG_{1}. \end{align}
 Noting the invariance of $h(X)$ given in \eqref{hPE} and integrating over $G_1$ and $U$, we immediately see that the joint probability distribution of $G_0$ and $X_0$  reads 
\begin{align}  \label{reducedmatrixpdf}
     \exp\big\{ -\alpha \tr(G_{0}^{}G_{0}^{*})+  \tr\big(\Omega G_0X_0+(\Omega G_0X_0)^*\big)\big\}
      \det(X_{0}^{*}X_{0}^{})^{M-N}\, h(X_0) dG_{0} dX_0   \end{align}
up to some constant. Furthermore, both $GX$ and $G_0 X_0$ have the same singular values.

\textbf{Step 2: Joint  singular value PDF  of $X$ and $Y$}. To get the squared singular values of the product $GX$ it suffices to study  the distribution defined in \eqref{reducedmatrixpdf}. For simplicity sake,  we replace the notation $G_0$ and $X_0$ with $G$ and $X$ respectively.

Since the change of variables of $G\mapsto Y=GX$ and $X\mapsto  X$ has a Jacobian $\det(X^{*}X^{})^{-L}$  where $X$ has  the full rank  $N$(cf. \cite[Theorem 3.2]{Mathai97}), $Y$ and $X$ have a joint distribution proportional to
\begin{align}  \label{YXmatrixpdf}
     \exp\big\{ -\alpha \tr\big(Y^{*}Y^{}(X^{*}X^{})^{-1}\big)+  \tr\big(\Omega Y+(\Omega Y)^*\big)\big\}
      \det(X^{*}X^{})^{M-N-L}\, h(X) dY dX.    \end{align}

Next, let  $\Lambda_{x}=\textrm{diag}\big(x_1, \ldots, x_N\big)$ and $\Lambda_{t}=\textrm{diag}\big(t_1, \ldots, t_N\big)$, according to  the singular value decomposition,  both $Y$ and $X$ can be written as
\be Y=U \begin{pmatrix}  \sqrt{\Lambda_{x}} \\ O\end{pmatrix}  V, \qquad X=W \sqrt{\Lambda_{t}}Q,\ee
where $U$ is an $L\times N$ complex matrix with $U^* U=I_N$, all $V, W$ and $Q$ are $N\times N$ unitary matrices. Then both the Jacobians  read 
\be dY\propto \prod_{k=1}^{N} x_{k}^{L-N}\Delta(x)^{2}dU dV dx_{1} \cdots dx_{N},\label{Jacobian1}\ee
and \be dX\propto  \Delta(t)^{2}dW dQ dt_{1} \cdots dt_{N},\label{Jacobian2}\ee
see e.g. \cite[Chapt. 3]{Fo10}. Together with  \eqref{Jacobian1} and \eqref{Jacobian2},   by the invariance of the Haar measure under the change $Q \mapsto QV$,  we know    that \eqref{YXmatrixpdf} is reduced to the distribution proportional to
\begin{align}  \exp&\big\{ -\alpha \tr\big(\Lambda_{x}Q^{-1}\Lambda^{-1}_{t} Q\big)+  \tr\big(\Omega U\sqrt{\Lambda_{x}}V+(\Omega U\sqrt{\Lambda_{x}}V)^*\big)\big\} \prod_{k=1}^{N} \big(x_{k}^{L-N} t_{k}^{M-N-L}\big) \nonumber\\
& \ \times  \Delta(x)^{2} \Delta(t) \det[f_{k}(t_j)]_{j,k=1}^N dU dV dWdQ  dx_{1} \cdots dx_{N}dt_{1} \cdots dt_{N}. \end{align}

We need to use the Harish-Chandra-Itzykson-Zuber integral formula (cf. \cite{HC57} and \cite{IZ80})
\be \int_{U(N)}e^{-\alpha \tr (\Lambda_{x}Q^{-1}\Lambda^{-1}_{t} Q )}dQ = C_{N}\frac{\det[e^{-\alpha x_{j}/t_{k}}]_{j,k=1}^N}{\Delta(x)\Delta(1/t)}
\ee
and its analogue (cf. \cite{GW96} and \cite{JSV97})
\begin{multline} \int_{\{U: U^* U=I_N\}} \int_{V\in U(N)}e^{-\tr\big(\Omega U\sqrt{\Lambda_{x}}V+(\Omega U\sqrt{\Lambda_{x}}V)^*\big)}dUdV  \\ =C_{L,N}\frac{\det[_0F_1(L-N+1;x_{j} \delta^{2}_{k})]_{j,k=1}^N}{\Delta(x)\Delta(\delta^2)},
\end{multline}
where $C_N$ depends only on $N$ and $C_{L,N}$ only on $L$ and $N$. Accordingly, integrate out $U, V, W, Q$ parts and note that $\Delta(1/t) =(-1)^{N(N-1)/2}\prod_{k=1}^{N} t_{k}^{1-N} \Delta(t)$,  we   obtain the joint distribution of squared singular values for $Y$ and $X$ which is proportional to
\begin{multline}   \det[e^{-\alpha x_{j}/t_{k}}]_{j,k=1}^N  \det[_0F_1(L-N+1;x_{j} \delta^{2}_{k})]_{j,k=1}^N  \det[t_{j}^{M-N}f_{k}(t_j)]_{j,k=1}^N\\ \times   \frac{1}{\Delta(\delta^2)}  \prod_{k=1}^{N} \big(\frac{x_{k}}{t_{k}}\big)^{L-N} \frac{dt_{1}}{t_1} \cdots \frac{dt_{N}}{t_N}dx_{1} \cdots dx_{N}. \label{jointPDFXY}\end{multline}

\textbf{Step 3: Singular value PDF of $Y$}. In order to derive  the joint PDF for the squared singular values of $Y$, we need to integrate out all variables $t_1, \ldots, t_N$ in \eqref{jointPDFXY}.    This can be done with the aid of the Andr\'{e}ief integral identity (see e.g. \cite[Sect. 3.1]{DG09}) so that
\begin{multline}   \int_{0}^{\infty}\cdots \int_{0}^{\infty}\det[e^{-\alpha x_{j}/t_{k}}]_{j,k=1}^N   \det[t_{j}^{M-N}f_{k}(t_j)]_{j,k=1}^N \prod_{k=1}^{N} \big(\frac{x_{k}}{t_{k}}\big)^{L-N} \frac{dt_{1}}{t_1} \cdots \frac{dt_{N}}{t_N} \\ = N! \det[\eta_{k}(x_j)]_{j,k=1}^N, \end{multline}
where
\begin{align} \eta_{k}(z)&= \int_{0}^{\infty}  e^{-\alpha \frac{z}{t}}   t^{M-N}f_{k}(t)\big(\frac{z}{t}\big)^{L-N} \frac{dt}{t}\nonumber \\
&= \int_{0}^{\infty} t^{L-N} e^{-\alpha t} \big(\frac{z}{t}\big)^{M-N} f_{k}\big(\frac{z}{t}\big) \frac{dt}{t}. \end{align}
This gives us the requested joint PDF \eqref{eigenvaluepdf}.

To evaluate the normalization constant, we make use of   the Andr\'{e}ief   identity again as follows
\begin{align} Z_N=N! \det\Big[\int_{0}^{\infty} {_0F_1}(L-N+1;x \delta^{2}_{j})\, \eta_{k}(x) dx\Big]_{j,k=1}^N. \end{align}
Change variables   $x\mapsto xt,t\mapsto t$,  integrate term by term in the inner integral and we then obtain
\begin{align}  &\int_{0}^{\infty}  {_0F_1}(L-N+1;x \delta^{2}_{j})\, \eta_{k}(x) dx  \nonumber \\
&=  \int_{0}^{\infty}  \int_{0}^{\infty} {_0F_1}(L-N+1;x \delta^{2}_{j})\, t^{L-N} e^{-\alpha t} \big(\frac{x}{t}\big)^{M-N} f_{k}\big(\frac{x}{t}\big) \frac{dt}{t} dx \nonumber \\
&=\int_{0}^{\infty} \Big( \int_{0}^{\infty} {_0F_1}(L-N+1;xt \delta^{2}_{j}) \,  t^{L-N} e^{-\alpha t}  dt\Big) x^{M-N} f_{k}(x) dx \nonumber\\
&=(L-N)!  \alpha^{-(L-N+1)} \int_{0}^{\infty} e^{x\delta_{j}^{2}/\alpha} x^{M-N}f_{k}(x) dx,\end{align}
from which the normalization constant follows. Here in the second identity above we have applied  the Fubini's theorem, since  the assumptions on  functions $f_{k}$ imply   $|f_{k}(x/t)|\leq C e^{-\alpha x/t}$  for some constant $C$.
 \end{proof}

We can apply Theorem \ref{coupledmultiplication} to any  bi-invariant random matrix ensemble  $X$     which can be coupled together with a Ginibre matrix and has a joint singular value PDF as in \eqref{hPE}.    A few  examples   immediately follow from the above theorem.
\begin{example} \label{Vpdf} For the joint PDF    \eqref{twomatrixpdf}, suppose that   $h(X)=\exp\{-\mathrm{Tr}V(X^{*}X)\}$  where $V$ is a polynomial  with positive leading coefficient and  $\delta_1, \ldots, \delta_N$ are singular values of $\Omega$. Then the  squared singular values    of $Y=GX$ has a joint PDF on $[0,\infty)^N$
 \be   \mathcal{P}_{N}(x_1,\ldots,x_N)=\frac{1}{Z_N} \det[\xi_i(x_j))]_{i,j=1}^{N} \det[\eta_i(x_j))]_{i,j=1}^{N}, 
 \ee
 where  $\xi_{i}(z)= {_0F_1}(L-N+1;\delta_{i}^{2}z)$,
  \be \eta_{i}(z)= \int_{0}^{\infty} t^{L-N} e^{-\alpha t -V(\frac{z}{t})}  \big(\frac{z}{t}\big)^{M-N+i-1}  \frac{dt}{t}, \ee
  and the normalization constant
  \be Z_N=N! ((L-N)!)^N \alpha^{-N(L-N+1)}\det\Big[\int_{0}^{\infty} x^{M-N+k-1}   e^{-V(x)+x\delta_{j}^{2}/\alpha} dx\Big]_{j,k=1}^N. \label{Zpartition}\ee
  \end{example}

Yet another family of random  matrix ensembles
with singularities of the form
\be  h(X)=\textrm{const} \cdot  \exp\{- \alpha \mathrm{Tr}(X^{*}X^{})-\beta^{d}\mathrm{Tr}(X^{*}X^{})^{-d}\}, \qquad  \beta>0 \ \text{and} \ d\in \mathbb{N},   \nonumber \ee
where $X$ is a    complex matrix  of size $M\times N$
 was studied    in \cite{ACM15,CI10,XDZ14} and   a hard edge limiting kernel  was obtained
 in terms of the Painlev\'{e} III hierarchy  \cite{ACM15,XDZ14}.   The singular value PDF for a coupled product with that reads as follows.

 \begin{example} \label{inversepdf} With     \eqref{twomatrixpdf}, let  $h(X)=\exp\{- \alpha \mathrm{Tr}(X^{*}X^{})-\beta^{d}\mathrm{Tr}(X^{*}X^{})^{-d}\}$  where  $\beta>0$ and $d\in \mathbb{N}$,  and  let $\delta_1, \ldots, \delta_N$ be singular values of $\Omega$. Then the  squared singular values    of $Y=GX$ has a joint PDF on $[0,\infty)^N$
 \be   \mathcal{P}_{N}(x_1,\ldots,x_N)=\frac{1}{Z_N} \det[\xi_i(x_j))]_{i,j=1}^{N} \det[\eta_i(x_j))]_{i,j=1}^{N}, \nonumber
 \ee
 where  $\xi_{i}(z)= {_0F_1}(L-N+1;\delta_{i}^{2}z)$,
  \be \eta_{i}(z)= \int_{0}^{\infty} t^{L-N} e^{-\alpha(t+\frac{z}{t})-\beta^{d}(\frac{z}{t})^{-d}}  \big(\frac{z}{t}\big)^{M-N+i-1}  \frac{dt}{t}, \nonumber \ee
  and the normalization constant
  \be Z_N=N! ((L-N)!)^N \alpha^{-N(L-N+1)}\det\Big[\int_{0}^{\infty}   x^{M-N+k-1} e^{-(\alpha-\delta_{j}^{2}/\alpha)x-(\beta/x)^{d}} dx\Big]_{j,k=1}^N. \label{Zpartition-2} \nonumber
  \ee
  \end{example}

We will get back to the random matrix ensembles stated in Examples  \ref{Vpdf} and \ref{inversepdf} in a forthcoming paper,  and expect similar hard edge  transition to occur as in Theorem \ref{hardlimits}, but  a detailed
study would lead us too far.

So far, we see that Proposition \ref{twoGpdf} is just a special case of  Example  \ref{Vpdf}.
\begin{proof}[Proof of Proposition \ref{twoGpdf}] Take $V(x)=\alpha x$ in Example  \ref{Vpdf}, recall \eqref{Kfunction} and we can  rewrite
\be \eta_{i}(z)= 2 (\sqrt{z})^{L+M-2N+i-1} K_{M-L+i-1}(2\alpha \sqrt{z}). \label{proofeq1}
 %\alpha^{M-L+i-1} z^{M-N+i-1}\int_{0}^{\infty} t^{-(M-L+i-1)} e^{- t - \frac{\alpha^2 z}{t} }  \frac{dt}{t},
\ee
Accordingly,   simple calculation shows that the constant given in \eqref{Zpartition}  is reduced to
\be Z_N=N! ((L-N)!)^N \alpha^{-N(L+M)-\frac{1}{2}N(N-1)} \Delta(\delta^2)
\prod_{j=1}^{N}\Big( \Gamma(M-N+j)\big(1-\frac{\delta_{j}^{2}}{\alpha^2} \big)^{-M}\Big). \label{proofeq2} \ee

 Combine \eqref{proofeq1} and \eqref{proofeq2},  recall \eqref{0F1I} and  we   have the desired result.  \end{proof}

\subsection{Coupled multiplication with a Jacobi matrix} \label{sectpdfJ}
Given two complex matrices   $X_1$ of size $(\kappa+N)\times (\nu +N)$  and  $X_2$ of size $(\nu +N)\times N$ with $\kappa, \nu\geq 0$, and a positive semidefinite $N\times N$ matrix  $\Sigma$,   we consider  the joint PDF which is proportional to
  \begin{align} \label{matrixpdfJ}  {_{1}\tilde{F}_{1}}(\nu+\nu'+2N;\kappa+N;X_1 X_2\Sigma   X_{2} ^* X_{1} ^*) \exp\!\big\{ -\alpha \tr  X_{2}^{*} X_{2}^{} \big\} \nonumber\\ \qquad \times {\det}^{\nu'-\kappa}{\Big(I_{\nu+N}-X_{1} X_{1}^{*} \Big)} \theta{\Big(I_{\nu+N}-X_{1} X_{1}^{*} \Big)} dX_1 d X_2,
  \end{align}
  where $\Sigma <\sqrt{\alpha} I_N$ and $\nu'$ is a non-negative integer such that $\nu+\nu'\geq \kappa$.
  Here ${_{1}\tilde{F}_{1}}$  is a hypergeometric function of matrix argument (see \cite{James64,GR89} for more details)  and the Heaviside step function of  matrix argument  defined on Hermitian matrices $H$ as
  \be \theta(H)=\begin{cases} 1, & \mbox{if $H$ is positive definite}, \\
   0, & \mbox{other}.
  \end{cases}  \nonumber \ee
This is expected to be closely related with the non-central distribution in MANOVA (see e.g. \cite{A03}), which may  be derived from the joint distribution  proportional to
  \begin{align}
  \exp\!\big\{ -  \tr(Z_{1}^{}Z_{1}^{*} + Z_{2}^{}Z_{2}^{*})\Sigma_{0}^{-1}  +\tr(\Omega Z_1 X_2+  (\Omega Z_1 X_2)^*  )
  -\alpha \tr  X_{2}^{*} X_{2}^{}\big\}, 
    \end{align}
  where $Z_1$ and $Z_2$ are rectangular matrices of sizes   $(\kappa+N)\times (\nu +N)$ and $(\kappa+N)\times (\nu' +N)$, respectively.  Set $X_1=(Z_{1}^{}Z_{1}^{*} + Z_{2}^{}Z_{2}^{*})^{-1/2}Z_{1}$, then $X_1$ and $X_2$ are expected to be distributed as in \eqref{matrixpdfJ} but with $\Sigma=\Omega \Sigma_0 \Omega^{*}$; cf. \cite[Sect. 8]{James64}.

Instead of \eqref{matrixpdfJ}, with the same notations as in \eqref{twomatrixpdf} we now turn to a more general PDF proportional to
\begin{align} \label{matrixpdfJ2}  {_{1}\tilde{F}_{1}}&(\nu+\nu'+2N;\kappa+N;GX\Sigma   X ^* G ^*) \, h(X)  \nonumber\\ \qquad &\times {\det}^{\nu'-\kappa}{\Big(I_{\nu+N}-G G^{*} \Big)} \theta{\Big(I_{\nu+N}-G G^{*} \Big)} dG dX,
  \end{align}
 where  $\Sigma <\sqrt{\alpha} I_N$ and $\nu,  \kappa,  \nu'$ are non-negative integers such that $\nu+\nu'\geq \kappa$.
  Likewise, the squared singular values of $GX$ also forms a determinantal point process.

\begin{theorem}\label{coupledmultiplicationJ} With the joint PDF defined in  \eqref{matrixpdfJ2}, let $\delta_{1}^{2}/\alpha, \ldots, \delta_{N}^{2}/\alpha$ be eigenvalues $\Sigma$ with $\alpha>0$ and all $0\leq \delta_j<\alpha$. Suppose that $f_{k}(t)$ ($k=1, \ldots, N$) are continuous in $(0,\infty)$ such that  all $e^{\alpha t}f_{k}(t)$  are bounded in $[0, \infty)$, let 
\be h(X)=\frac{1}{\Delta(t)}\det[f_{k}(t_j)]_{j,k=1}^N, \quad 0<t_1, \ldots, t_N<\infty, \label{hPEJ} \nonumber
\ee
 where  $t_1, \ldots, t_N$ are    eigenvalues of $X^{*} X^{}$, then the  squared singular values    of $Y=GX$ have a joint PDF on $[0,\infty)^N$
 \be \label{eigenvaluepdfJ} \mathcal{P}_{N}(x_1,\ldots,x_N)=\frac{1}{Z_N} \det[\xi_i(x_j))]_{i,j=1}^{N} \det[\eta_i(x_j))]_{i,j=1}^{N},    \nonumber  \ee
 where  $\xi_{i}(z)= z^{\kappa} {_1F_1}(\nu+\nu'+N+1;\kappa+1;\delta_{i}^{2}z/\alpha)$ and
  \be \eta_{i}(z)= \int_{0}^{1} (1-t)^{\nu+\nu'-\kappa+N-1}    \big(\frac{z}{t}\big)^{\nu-\kappa} f_{i}\big(\frac{z}{t}\big) \frac{dt}{t}.  \nonumber\ee
The normalization constant can be evaluated by
  \be Z_N=N! \bigg(\frac{\Gamma(\kappa+1) \Gamma(\nu+\nu'-\kappa+N)}{\Gamma(\nu+\nu'+N+1)}\bigg) ^N  \det\Big[\int_{0}^{\infty} x^{\nu} e^{x\delta_{j}^{2}/\alpha} f_{k}(x) dx\Big]_{j,k=1}^N. \nonumber\ee
  \end{theorem}

A corollary    immediately follows from the above theorem.
\begin{corollary} \label{pdfJ} For the joint PDF    \eqref{matrixpdfJ2}, let  $h(X)=\exp\{-\alpha\mathrm{Tr}(X^{*}X) \}$. With the same notations as in Theorem \ref{coupledmultiplicationJ}, then the  squared singular values    of $Y=GX$ have a joint PDF on $[0,\infty)^N$
 \be   \mathcal{P}_{N}(x_1,\ldots,x_N)=\frac{1}{Z_N} \det[\xi_i(x_j))]_{i,j=1}^{N} \det[\eta_i(x_j))]_{i,j=1}^{N}, \label{BiEJ}
 \ee
 where  $\xi_{i}(z)= z^{\kappa} {_1F_1}(\nu+\nu'+N+1;\kappa+1;\delta_{i}^{2}z/\alpha)$,
  \be \eta_{i}(z)= \int_{0}^{1} (1-t)^{\nu+\nu'-\kappa+N-1}    \big(\frac{z}{t}\big)^{\nu-\kappa+i-1} e^{-\alpha\frac{z}{t}} \frac{dt}{t}, \ee
with  $0\leq \delta_{j}<\alpha$ for $j=1, \ldots, N$ and the normalization constant
  \begin{multline} Z_N=N! \alpha^{-\frac{1}{2}N(N-1)}\Delta(\delta^2)   \\ \times \bigg(\frac{\Gamma(\kappa+1) \Gamma(\nu+\nu'-\kappa+N)}{\Gamma(\nu+\nu'+N+1)}\bigg) ^N  \prod_{j=1}^{N}\bigg(\Gamma(\nu+j)   \Big(\alpha-  \frac{\delta_{j}^{2}}{\alpha} \Big)^{-\nu-N}\bigg).\end{multline}
  \end{corollary}

\begin{proof}[Proof of Theorem \ref{coupledmultiplicationJ}]   Since we can proceed almost in the same steps as in Theorem \ref{coupledmultiplication}, we just point out some different places and leave the details to the reader.

\textbf{Step 1: Reduction}.
Let us assume that $\nu>0$. Since  any matrix $X$ of size $M\times N$ can be decomposed as
\be X=U\binom{X_0}{O}, \nonumber
\ee
where $U$ is a  $(\nu+N)\times (\nu+N)$ unitary matrix  and $X_0$ is an $N\times N$ complex matrix, setting $GU=(G_0\ G_1)$ we arrive at the joint distribution of $G_0, G_1,  X_0$ and $U$   proportional to
\begin{multline*}  %\label{threematrixpdfJ}
{_{1}\tilde{F}_{1}}(\nu+\nu'+2N;\kappa+N;G_{0}X_{0}\Sigma   X_{0}^{*} G_{0}^{*})  \, h(X_0) \, {{\det}^{\nu}(X_{0}^{*}X_{0}^{})}   \\   \times  {\det}^{\nu'-\kappa}{\Big(I-G_{0} G_{0}^{*}-G_{1} G_{1}^{*} \Big)}
 \theta{\Big(I-G_{0} G_{0}^{*}-G_{1} G_{1}^{*} \Big)}
 \,   dG_{0} dG_{1} dX_{0} [dU]. 
\end{multline*}

Make a change of variables $G_1\mapsto \Big(I-G_{0} G_{0}^{*}\Big)^{1/2} G_1$ and    integrate over $G_1$ and $U$, we immediately see that the joint probability distribution of $G_0$ and $X_0$  reads 
\begin{multline}  \label{reducedmatrixpdfJ}
{_{1}\tilde{F}_{1}}(\nu+\nu'+2N;\kappa+N;G_{0}X_{0}\Sigma   X_{0}^{*} G_{0}^{*})  \, h(X_0) \,{{\det}^{\nu}(X_{0}^{*}X_{0}^{})}   \\   \times  {\det}^{\nu'+\nu-\kappa}{\Big(I-G_{0} G_{0}^{*} \Big)}\,
 \theta{\Big(I-G_{0} G_{0}^{*}  \Big)}
 \,   dG_{0}   dX_{0}.
\end{multline}
up to some constant. Furthermore, both $GX$ and $G_0 X_0$ have the same singular values, which shows that    we  only need to  focus on the joint distribution  \eqref{reducedmatrixpdfJ}.

\textbf{Step 2: Joint  singular value PDF  of $X$ and $Y$}. For convenience,    we  next  replace the notation $G_0$ and $X_0$ with $G$ and $X$ respectively in \eqref{reducedmatrixpdfJ}.
Since the change of variables of $G\mapsto Y=GX$ and $X\mapsto  X$ has a Jacobian $\det(X^{*}X^{})^{-L}$ (cf. \cite[Theorem 3.2]{Mathai97}), $Y$ and $X$ have a joint distribution proportional to
\begin{multline}  \label{YXmatrixpdfJ}
{_{1}\tilde{F}_{1}}(\nu+\nu'+2N;\kappa+N; Y\Sigma   Y^{*})  \, h(X)\, {{\det}^{\nu-\kappa-N}(X^{*}X^{})}   \\   \times  {\det}^{\nu'+\nu-\kappa}{\Big(I-(X^{*}X)^{-1}Y^{*} Y\Big)}  \,  \theta{\Big(X^{*}X-Y^{*} Y  \Big)} dY   dX.
\end{multline}

Next, let  $\Lambda_{x}=\textrm{diag}\big(x_1, \ldots, x_N\big)$, according to  the singular value decomposition,  write
\be Y=U \begin{pmatrix}  \sqrt{\Lambda_{x}} \\ O\end{pmatrix}  V, \qquad X=W \sqrt{\Lambda_{t}}Q, \nonumber \ee
 note the fact ${_{1}\tilde{F}_{1}}(\cdot;\cdot; Y\Sigma   Y^{*})={_{1}\tilde{F}_{1}}(\cdot;\cdot;  \Sigma   Y^{*} Y)$  and change $Q \mapsto QV$,  due to the invariance of the Haar measure  we know    that \eqref{YXmatrixpdfJ} is reduced to the distribution proportional to
\begin{align} &{_{1}\tilde{F}_{1}}(\cdot;\cdot;  \Sigma V^{-1}  \Lambda_{x} V) \,  {\det}^{\nu'+\nu-\kappa}{\Big(I-\Lambda^{-1}_{t} Q\Lambda_{x}Q^{-1}\Big)} \,  \theta{\Big(\Lambda_{t}- Q\Lambda_{x}Q^{-1}\Big)}    \prod_{k=1}^{N} \big(x_{k}^{\kappa} t_{k}^{\nu-\kappa-N}\big) \nonumber\\
& \ \times  \Delta(x)^{2} \Delta(t) \det[f_{k}(t_j)]_{j,k=1}^N dU dV dWdQ  dx_{1} \cdots dx_{N}dt_{1} \cdots dt_{N}. \nonumber\end{align}

  We need to use the following two    integral formulas  over the unitary group
\begin{multline}   \int_{ U(N)}{_{1}\tilde{F}_{1}}(\nu+\nu'+2N;\kappa+N;  \Sigma V^{-1}  \Lambda_{x} V) dV \propto \\   \frac{1}{\Delta(x)\Delta(\delta^2)} \det\!\left[_1F_1(\nu+\nu'+N+1;\kappa+1;x_{j} \delta^{2}_{k}/\alpha)\right]_{j,k=1}^N, 
\end{multline}
where the ordinary hypergeometric function appears inside the determinant,  
and
\begin{multline}   \int_{U(N)}    {\det}^{\nu'+\nu-\kappa}{\Big(I-\Lambda^{-1}_{t} Q\Lambda_{x}Q^{-1}\Big)} \, \theta{\Big(\Lambda_{t}- Q\Lambda_{x}Q^{-1}\Big)}  dQ   \propto \\   \frac{1}{\Delta(x)\Delta(1/t)} \det\!\Big[ \Big(1-\frac{x_{j}}{ t_{k}} \Big)^{\nu'+\nu-\kappa+N-1}_{+}\Big]_{j,k=1}^N,
\end{multline}
where $x_{+}=\max\{0, x\}$; see  e.g.   \cite[Sect. 4]{GR89} or  \cite[Sect. 2]{Liu14} for the first formula, and  \cite[Theorem 2.3]{KKS15} for the later.
  Accordingly, integrate out $U, V, W, Q$ parts and note that $\Delta(1/t) =(-1)^{N(N-1)/2}\prod_{k=1}^{N} t_{k}^{1-N} \Delta(t)$,  we  thus arrive at the joint distribution of squared singular values for $Y$ and $X$   proportional to
\begin{multline}     \det\!\left[x_{j}^{\kappa} {_1F_1(\nu+\nu'+N+1;\kappa+1; x_{j} \delta^{2}_{k}/\alpha)}\right]_{j,k=1}^N
 \det[t_{j}^{\nu-\kappa-1}f_{k}(t_j)]_{j,k=1}^N
 \\ \times     \det\!\Big[ \Big(1-\frac{x_{j}}{ t_{k}} \Big)^{\nu'+\nu-\kappa+N-1}_{+}\Big]_{j,k=1}^N   .
 \label{jointPDFXYJ}\end{multline}

\textbf{Step 3: Singular value PDF of $Y$}. Integrating out all variables $t_1, \ldots, t_N$ in \eqref{jointPDFXYJ} and  using the Andr\'{e}ief integral identity,
 we have  the requested joint PDF   after some simple manipulations. To evaluate the normalization constant, we make use of   the Andr\'{e}ief   identity again to obtain  
\begin{align} Z_N=N! \det\Big[\int_{0}^{\infty}
x^{\kappa} {_1F_1}(\nu+\nu'+N+1;\kappa+1;\delta_{j}^{2}x/\alpha)\, \eta_{k}(x) dx\Big]_{j,k=1}^N. \nonumber \end{align}
Change variables   $x\mapsto xt,t\mapsto t$,  integrate term by term in the inner integral and we then  get for $a:=\nu+\nu'+N$
\begin{align}  &\int_{0}^{\infty}  x^{\kappa} {_1F_1}(\nu+\nu'+N+1;\kappa+1;\delta_{j}^{2}x/\alpha) \, \eta_{k}(x) dx  \nonumber \\
&=  \int_{0}^{\infty}  \int_{0}^{1}  x^{\kappa} {_1F_1}(a+1;\kappa+1;\delta_{j}^{2}x/\alpha)\, (1-t)^{a-\kappa-1}  \big(\frac{x}{t}\big)^{\nu-\kappa} f_{k}\big(\frac{x}{t}\big) \frac{dt}{t} dx \nonumber \\
&=\int_{0}^{\infty} \Big( \int_{0}^{1}    t^{\kappa}(1-t)^{a-\kappa-1}  {_1F_1}(a+1;\kappa+1;\delta_{j}^{2}xt/\alpha) dt\Big) x^{\nu} f_{k}(x) dx \nonumber\\
&=\frac{\Gamma(\kappa+1) \Gamma(\nu+\nu'-\kappa+N)}{\Gamma(\nu+\nu'+N+1)}  \int_{0}^{\infty} e^{x\delta_{j}^{2}/\alpha} x^{\nu}f_{k}(x) dx, \label{calculationJ}\end{align}
from which the desired normalization constant follows.  Here in the second identity above we have applied  the Fubini's theorem, since  the assumptions on  functions $f_{k}$ imply   $|f_{k}(x/t)|\leq C e^{-\alpha x/t}$  for some constant $C$.
 \end{proof}

\section{Double integrals  for correlation kernels} \label{sectkernel}

As mentioned earlier, the joint eigenvalue density \eqref{eigenvaluepdf} is an example of biorthogonal ensembles in Borodin's sense   \cite{Bo98}
  \be \mathcal{Q}_{N}(x_1,\ldots,x_N)=\frac{1}{Z_N} \det[\eta_i(x_j)]_{i,j=1}^{N}\det[\xi_i(x_j)]_{i,j=1}^{N}, \qquad  x_1, \ldots, x_N \in I,\label{bi-ensemblespdf}\ee
where $I$ is a union of finite intervals of $\mathbb{R}$.  The significance of the structure  \eqref{bi-ensemblespdf}  is that there exists a systematic way to compute the corresponding $k$-point correlation functions defined by
\be \rho_{k}(x_1,\dots,x_k)=\frac{N!}{(N-k)!}\int \cdots \int \mathcal{Q}_{N}(x_1,\ldots,x_N)\, d x_{k+1} \cdots d x_{N}, \nonumber \ee see e.g. \cite[Eq.(5.1)]{Fo10}.
The following proposition due to Borodin provides a solution to derive the correlation kernel which is of vital  importance in the study of determinantal point processes.
\begin{prop}[{\cite[Proposition~2.2]{Bo98}}] \label{PB}
Let $g_{i,j} :=  \int  \eta_i(x) \xi_j(x) \, dx$, suppose that  $[g_{i,j}]_{i,j=1}^{n}$ be invertible for each
$n=1,2,\dots$. Defining $c_{i,j}$ by
\begin{equation}\label{7.eb1'}
\big([c_{i,j}]_{i,j=1}^{N}\big)^t =
\big ( [g_{i,j}]_{i,j=1}^{N} \big )^{-1},
\end{equation}
and  setting \begin{equation}\label{7.eb2'}
K_N(x,y) =  \sum_{i,j=1}^N c_{i,j} \eta_i(x) \xi_j(y),
\end{equation}
  we then have
\begin{equation}
\rho_{k}(x_1,\dots,x_k) =
\det [ K_N(x_j,x_l) ]_{j,l=1}^{k}. \nonumber
\end{equation}
\end{prop}

Next, we first use Proposition \ref{PB} to complete the proof of Theorem   \ref{kernelpdf}.% cf. \cite{DF08}, \cite{FL15} or  \cite{BBP}.
\begin{proof}[Proof of Theorem \ref{kernelpdf}] Starting with the eigenvalue PDF \eqref{eigenvaluepdfGinibre}, with   \eqref{bi-ensemblespdf} in mind  we set
\be \eta_i(x)= x^{\frac{\nu+i-1}{2}} K_{\nu-\kappa+i-1}(2\alpha\sqrt{x}),  \qquad \xi_i(x)=I_{\kappa}(2\delta_i\sqrt{x}). \label{explicitform} \ee
In order to  calculate the integral $g_{i,j}$ as  presented in Proposition \ref{PB}, we   make  use of the integral formula involving Bessel functions as in \cite{AS15},
 \be    \int_{0}^{\infty} t^{\mu+\nu+1} K_{\mu}(a t) I_{\nu}(b t)  dt=\frac{2^{\mu+\nu}b^{\nu}\Gamma(\mu+\nu+1)}{ a^{\mu+2\nu+2}}  (1-\frac{b^2}{a^2})^{-\mu-\nu-1}, \label{besselintegral} \ee
which  can be derived by applying Euler's transformation for Gaussian hypergeometric functions $_{2}F_{1}$ in    \cite[6.576.5]{GR07} and then by taking a special case (noting that  we have  in essence given a  direct derivation for the formula \eqref{besselintegral} in the proof of Proposition \ref{twoGpdf}; cf.  \eqref{proofeq2}). Here $a>b>0$, and $\mu+\nu+1>0$.  Then by \eqref{besselintegral} we have
\begin{align}
g_{i,j}&=\int_0^\infty x^{\frac{\nu+i-1}{2}} K_{\nu-\kappa+i-1}(2\alpha\sqrt{x})\, I_{\kappa}(2\delta_j\sqrt{x}) \, dx \nonumber\\
&= \Big(1-\frac{\delta_{j}^2}{\alpha^2}\Big)^{-\nu-N} \frac{\Gamma(\nu+i)\delta_{j}^{\kappa}}{2\alpha^{\nu+\kappa+i+1}} \Big(1-\frac{\delta_{j}^2}{\alpha^2}\Big)^{N-i}.
\end{align}

According to Proposition \ref{PB}, with $G=[g_{i,j}]_{i,j=1}^N$, let $C=(G^{-1})^{t}$, the entries $c_{i,j}$ of $C$ then satisfy
\be \label{bv}  \sum_{i=1}^{N}   \Big(1-\frac{\delta_{k}^2}{\alpha^2}\Big)^{-\nu-N} \frac{\Gamma(\nu+i)\delta_{k}^{\kappa}}{2\alpha^{\nu+\kappa+i+1}} \Big(1-\frac{\delta_{k}^2}{\alpha^2}\Big)^{N-i}  \  c_{i,j}=\delta_{j,k}, \nonumber\ee
that is,
\be \label{bv2}  \sum_{i=1}^{N}    \frac{\Gamma(\nu+i)}{2\alpha^{\nu+\kappa+i+1}} \Big(1-\frac{\delta_{k}^2}{\alpha^2}\Big)^{N-i}  \  c_{i,j}=\delta_{k}^{-\kappa}\Big(1-\frac{\delta_{k}^2}{\alpha^2}\Big)^{\nu+N}\delta_{j,k}.\ee
Without loss of generality, we assume that $\delta_1, \ldots, \delta_N$ are pairwise distinct; otherwise, the requested result follows from taking proper limit and using of L'Hospital's rule. In this case, the above equations imply
 \be \sum_{i=1}^{N}   \frac{\Gamma(\nu+i)}{2\alpha^{\nu+\kappa+i+1}}  \, u^{N-i}  \,  c_{i,j}  =   \delta_{j}^{-\kappa} \Big(1-\frac{\delta_{j}^2}{\alpha^2}\Big)^{\nu+N} \prod_{l=1,l\neq j}^{N} \frac{u-\big(1-\frac{\delta_{l}^2}{\alpha^2}\big)}{\big(1-\frac{\delta_{j}^2}{\alpha^2}\big)-\big(1-\frac{\delta_{l}^2}{\alpha^2}\big)}, \label{csumidentity}\ee
as can be verified by noting that both sides are polynomials of degree $N-1$ in $u$ and take the same values at $N$ different points since  \eqref{bv2}  holds true.

Using this implicit formula for $\{c_{i,j}\}$ we are ready to   show that  \eqref{7.eb2'}  implies the double contour integral
formula  \eqref{kernelCD}. Keep  \eqref{Kfunction} in mind and also  note that for a positive integer $l$ (cf. Hankel's formula for the reciprocal gamma function)
\be z^{l-1}=\frac{\Gamma(l)}{ 2\pi i }\int_{\mathcal{C}_{\mathrm{out}}} u^{-l} e^{zu}du, \label{inversegamma}\ee 
we have from  \eqref{7.eb2'}  that
  \begin{align} K_N(x,y) &=  \sum_{j=1}^{N}   \xi_j(y)   \sum_{i=1}^{N} \frac{1}{2}x^{\frac{\kappa}{2}}   \int_{0}^{\infty} \frac{dt}{t} \, \big(\frac{\alpha x}{t}\big)^{i+\nu-1-\kappa} e^{-t-\frac{\alpha^2 x}{t}}   \, c_{i,j}\nonumber \\
     &=  \sum_{j=1}^{N}   \xi_j(y)    \frac{1}{2}x^{\frac{\kappa}{2}}  \int_{0}^{\infty} \frac{dt}{t} \big(\frac{x}{t}\big)^{-\kappa} e^{-t-\frac{\alpha^2 x}{t}} \nonumber \\
     & \quad \times \frac{2\alpha^2}{2\pi i}  \sum_{i=1}^{N} \int_{\mathcal{C}_{\mathrm{out}}}  du  \, e^{\frac{\alpha^2 x}{t}u}   \frac{\Gamma(\nu+i)}{2\alpha^{\nu+\kappa+i+1}} u^{-\nu-i} \, c_{i,j}, \label{utfunction}
   \end{align}
   where   the simple closed contour  $\mathcal{C}_{\mathrm{out}}$ is chosen such that it  doesn't   depend on any parameters $\alpha, \delta_1, \ldots, \delta_N$ 
   and  $\mathrm{Re}(z)<1$ for $z\in \mathcal{C}_{\mathrm{out}}$. Then for  $x>0$  the integrand with variables $u,t $ on the RHS of the second identity of   \eqref{utfunction} permits us to exchange the order of integration. Combine   the identity  \eqref{csumidentity} and  we thus get  
    \begin{align} K_N(x,y) &=  \frac{2\alpha^2}{2\pi i}\sum_{j=1}^{N}   \xi_j(y)   \delta_{j}^{-\kappa} \Big(1-\frac{\delta_{j}^2}{\alpha^2}\Big)^{\nu+N}  \int_{\mathcal{C}_{\mathrm{out}}}  du  \,    u^{-\nu-N} \nonumber \\
     & \quad \times  \prod_{l=1,l\neq j}^{N} \frac{u-\big(1-\frac{\delta_{l}^2}{\alpha^2}\big)}{\big(1-\frac{\delta_{j}^2}{\alpha^2}\big)-\big(1-\frac{\delta_{l}^2}{\alpha^2}\big)}
     \, \frac{1}{2}x^{-\frac{\kappa}{2}}  \int_{0}^{\infty} \frac{dt}{t} t^{\kappa} e^{-t-\frac{\alpha^2 x}{t}(1-u)}.  \label{sumj}
  \end{align}

  Finally,   recall \eqref{Kfunction} and \eqref{explicitform}, we rewrite the  summation in \eqref{sumj}   as
\begin{align}
  K_N(x,y)  = \frac{2\alpha^2}{2\pi i} \sum_{j=1}^{N} &  I_{\kappa}(2\delta_j\sqrt{ y})   \delta_{j}^{-\kappa} \Big(1-\frac{\delta_{j}^2}{\alpha^2}\Big)^{\nu+N}        \int_{\mathcal{C}_{\mathrm{out}}}  du   \, u^{-\nu-N}  (\alpha\sqrt{1-u})^{\kappa} \nonumber \\
    & \times   K_{-\kappa}(2\alpha\sqrt{(1-u)x}) \prod_{l=1,l\neq j}^{N} \frac{u-\big(1-\frac{\delta_{l}^2}{\alpha^2}\big)}{\big(1-\frac{\delta_{j}^2}{\alpha^2}\big)-\big(1-\frac{\delta_{l}^2}{\alpha^2}\big)}.
\end{align}
  We recognise the above summation over $j$ as the  summation  of the residues at $\{1-\delta_{j}^{2}/\alpha^2\}$ of the $v$-function
\be v^{\nu+N} (\alpha\sqrt{1-v})^{-\kappa} I_{\kappa}(2\alpha\sqrt{(1-v)y})  \frac{1}{u-v} \prod_{l=1}^{N}\frac{u-(1-\delta_{l}^{2}/\alpha^2)}{v-(1-\delta_{l}^{2}/\alpha^2)} ,\ee
  application of the residue theorem then gives   the required result.  Here $\mathcal{C}_{\mathrm{in}}$ is a counterclockwise contour encircling $1-\delta_{1}^{2}/\alpha^2,\ldots, 1-\delta_{N}^{2}/\alpha^2$ but not any $u\in \mathcal{C}_{\mathrm{out}}$. In particular, we can choose the two contours as described in the theorem.  
\end{proof}

Note that in order to derive the double contour integral in  Theorem \ref{kernelpdf} we have made the best of  nice formulas for integrals  of Bessel functions, a question arises naturally: Are there double contour integrals for correlation kernels of the bi-orthogonal ensembles   \eqref{eigenvaluepdf}?  And even more specifically, is there a relationship between the correlation kernels associated with singular values of $GX$ and $X$? When $G$ and $X$ are independent, for $G$ being a Ginibre or truncated unitary matrix,  Claeys, Kuijlaars and Wang found a nice relation; see  \cite[Lemma 2.14]{CKW15}. It is really a challenge for us to extend their result to the coupled product case.

Secondly,  we have a Jacobi-type  analogue of  Theorem \ref{kernelpdf}. For this purpose, we need to define two   functions which can be treated as being  of mutual duality for an integral representation of correlation kernel. One is, as an entire function of $z$,  for $\nu>\kappa >-1$,  
\begin{align} f_1(\nu, \kappa; z)&=\frac{ \Gamma(\nu+1)}{\Gamma(\kappa+1) \Gamma(\nu-\kappa)}{_1F_1(\nu+1;\kappa+1; z)} \label{deff1}\\
&=\frac{ \Gamma(\nu+1)}{\Gamma(\nu-\kappa)}\frac{1}{2\pi i}  \int_{\mathcal{C}_{\mathrm{0}}}    \,       s^{-\kappa-1}
 e^{s}\left(1-\frac{z}{s}\right)^{-\nu-1}ds,\label{deff1-1}\end{align}
where $\mathcal{C}_{\mathrm{0}}$ is a   counterclockwise contour around the origin. The other is,  for $\nu>0$ and $\kappa\in \mathbb{R}$,  
\be f_2(\nu, \kappa; z)=\int_{0}^{1} (1-t)^{\nu-1}t^{\kappa-1}e^{-\frac{z}{t}}dt, \qquad |\textrm{arg}(z)|<\frac{\pi}{2}. \label{deff2} \ee
\begin{theorem}\label{kernelpdfJ}  
%%With the joint  PDF defined in  \eqref{matrixpdfJ}, let  $Y_2=X_1 X_2$. Then the joint eigenvalue  PDF    for $Y_{2}^{*} Y_{2}^{}$ can be  written in the form
%% \be \mathcal{P}_{N}(x_1,\ldots,x_N)=\frac{1}{N!} \det[K_N(x_i,x_j)]_{i,j=1}^{N},
%% \ee
%%with correlation kernel
The correlation kernel  for the biorthogonal ensemble  \eqref{BiEJ} is given by 
 \begin{multline}
K_N(x,y)=\frac{ \alpha }{(2\pi i)^2}\Big(\frac{y}{x}\Big)^{\kappa}\int_{\mathcal{C}_{\mathrm{out}}} du \int_{\mathcal{C}_{\mathrm{in}}} dv\,  f_{2}(\nu+\nu'-\kappa+N, \kappa; \alpha (1-u)x)
   \\  \times   {f_1(\nu+\nu'+N, \kappa; \alpha (1-v)y)}\, \frac{1}{u-v} \Big(\frac{u}{v}\Big)^{-\nu-N}\prod_{l=1}^{N}\frac{u-(1-\delta_{l}^{2}/\alpha^2)}{v-(1-\delta_{l}^{2}/\alpha^2)}, \label{kernelCDJ}\end{multline}
where  $\mathcal{C}_{\mathrm{in}}$ is a counterclockwise contour encircling $1-\delta_{1}^{2}/\alpha^2,\ldots, 1-\delta_{N}^{2}/\alpha^2$, and $\mathcal{C}_{\mathrm{out}}$   is a   simple   counterclockwise  contour around the origin with  $\mathrm{Re}(z)<1$ for $z\in \mathcal{C}_{\mathrm{out}}$  such that       
   $\mathcal{C}_{\mathrm{in}}$ is entirely  to the right side of  $\mathcal{C}_{\mathrm{out}}$.  When $0< \delta_{j}<\alpha$ for $j=1, \ldots, N$, we can also choose contours such that  
   $\mathcal{C}_{\mathrm{in}}$ is  contained entirely in  $\mathcal{C}_{\mathrm{out}}$.
 \end{theorem}

\begin{proof} We proceed in a similar way as in  Theorem \ref{kernelpdf} and just give a brief derivation as follows. With   Corollary \ref{pdfJ}  in mind,  simple calculation  in the same way as in \eqref{calculationJ} shows us
\begin{align}
g_{i,j} :&=\int_0^\infty \eta_{i}(x)\xi_{j}(x) dx   \nonumber\\
&= \frac{\Gamma(\kappa+1) \Gamma(\nu+\nu'-\kappa+N)}{\Gamma(\nu+\nu'+N+1)}  \frac{\Gamma(\nu+i) }{ \alpha^{\nu+i}} \Big(1-\frac{\delta_{j}^2}{\alpha^2}\Big)^{-\nu-i}.
\end{align}

According to Proposition \ref{PB}, with $G=[g_{i,j}]_{i,j=1}^N$, let $C=(G^{-1})^{t}$, the entries $c_{i,j}$ of $C$ then satisfy
   identical equations
 \begin{multline} \sum_{i=1}^{N}   \frac{\Gamma(\nu+i)}{2\alpha^{\nu+\kappa+i+1}}  \, u^{N-i}  \,  c_{i,j}  =   \frac{\Gamma(\nu+\nu'+N+1)}{\Gamma(\kappa+1) \Gamma(\nu+\nu'-\kappa+N)}  \\
 \times \Big(1-\frac{\delta_{j}^2}{\alpha^2}\Big)^{\nu+N} \prod_{l=1,l\neq j}^{N} \frac{u-\big(1-\frac{\delta_{l}^2}{\alpha^2}\big)}{\big(1-\frac{\delta_{j}^2}{\alpha^2}\big)-\big(1-\frac{\delta_{l}^2}{\alpha^2}\big)}. \label{csumidentityJ}\end{multline}
By   Hankel's formula \eqref{inversegamma}, we have from  \eqref{7.eb2'}  that
  \begin{align} K_N(x,y) &= \frac{\alpha}{2\pi i} \sum_{j }   \xi_j(y)       \int_{0}^{1} \frac{dt}{t}
   \big(\frac{x}{t}\big)^{-\kappa} (1-t)^{\nu+\nu'-\kappa+N-1}    e^{-\alpha\frac{x}{t}}  \nonumber \\
  & \quad \times \sum_{i }   \frac{\Gamma(\nu+i)}{ \alpha^{\nu+i}}\int_{\mathcal{C}_{\mathrm{out}}}  du  \, e^{\frac{\alpha x}{t}u}    u^{-\nu-i} \, c_{i,j}.
   \end{align}
 Thus, using  the identity  \eqref{csumidentityJ},  exchanging the order of integration and  then applying residue theorem imply  the required result.
\end{proof}

Finally,  let's  extract  some key ideas behind the proofs of both Theorems \ref{kernelpdf} and \ref{kernelpdfJ} and draw a general procedure in giving double contour integrals for correlation kernels in a class of bi-orthogonal ensembles; see \cite{BBP,CKW15,DF08,FL15b} for relevant examples.  
\begin{remark} \label{generldoubleI}For the bi-orthogonal ensemble  \eqref{bi-ensemblespdf},   suppose that the following conditions hold true:
\begin{itemize}
\item[(i)] There exist two functions $g(t,x)$,  $\Phi(t,x)$ and $N$ generic parameters $a_1, \ldots, a_N$ such that
\be \eta_i(x)=\int (xt)^{i-1} g(t,x) dt,  \qquad  \xi_j(x)=\Phi(a_j,x), \qquad x\in I.  \label{generaleq1}\ee
 \item[(ii)] There exist $h(x), q(x)$ and   polynomials  $L_{k}(x)$ of degree $k$ such that   \be   \int_{I} \eta_{i}(x)\xi_{j}(x) dx =\frac{1}{b_i h(a_j)}L_{i-1}(a_j)\label{generaleq2-1}\ee
where  \be b_i=\int z^{i-1} q(z)dz.\label{generaleq32-2}\ee
\item[(iii)] There exist $\tilde{g}(t,x)$ and a contour $\mathcal{C}_1$ not containing    $\{a_j\}$ such that      \be   z^{i-1}=\int_{\mathcal{C}_1} \frac{1}{h(u)} L_{i-1}(u) \tilde{g}(u,z) du. \label{generaleq3}\ee
 \end{itemize}
 Setting \be \Psi(u,x)= \iint \tilde{g}(u,xzt)q(z) g(t,x)dz dt,\label{generaleq4}\ee
 if both $h(z)$ and $\Phi(z,x)$ are analytic functions of $z$ in some proper domain  containing all $a_j$, then with certain conditions  such as integrability on $g, \tilde{g}, \Psi$ and $q$  the correlation kernel should be given by 
\be
K_N(x,y)=\frac{1 }{ 2\pi i } \int_{\mathcal{C}_{1}} du \int_{\mathcal{C}_{2}} dv\, \Psi(u,x) \Phi(v,y)\frac{h(v)}{h(u)}
   \frac{1}{u-v}  \prod_{l=1}^{N}\frac{u- a_l} {v-a_l}, \label{kernelCDge}\ee
where  $\mathcal{C}_{2}$ is a counterclockwise contour encircling $a_1,\ldots, a_N$, but does not intersect with $\mathcal{C}_{1}$.   
The integral transform  \eqref{generaleq3} connects polynomials  $L_{i}(z)$ and $z^{i}$, and  in a practical application of Proposition \ref{generldoubleI}  the most difficult part usually lies in finding  of a  suitable kernel  function  $\tilde{g}(t,x)$ as stated in Condition (iii).

 \end{remark}

\section{Hard edge limits}\label{sectionhardlimit}
In this section we are devoted to  the proof of  Theorem
\ref{hardlimits} and hard edge limits  of the kernel \eqref{kernelCDJ}, for which the same hard edge  transition phenomenon is observed.

 \begin{proof}[\textbf{Proof of Theorem \ref{hardlimits}}]
 First, under the assumptions   \eqref{finiterank} we can rewrite the correlation kernel  \eqref{kernelCD} as
  \begin{multline}
K_N(x,y)=\frac{2\alpha^2}{(2\pi i)^2}\int_{\mathcal{C}_{\mathrm{out}}} du \int_{\mathcal{C}_{\mathrm{in}}} dv\,K_{-\kappa}(2\alpha\sqrt{(1-u)x})\,  I_{\kappa}(2\alpha\sqrt{(1-v)y}) \\  \times \frac{1}{u-v} \Big(\frac{1-u}{1-v}\Big)^{\kappa/2}\Big(\frac{v}{u}\Big)^{\nu+m} \bigg(\frac{1-\frac{1}{u} \frac{4\mu}{(1+\mu)^2}}{1-\frac{1}{v}\frac{4\mu}{(1+\mu)^2}}\bigg)^{N-m} \prod_{l=1}^{m}\frac{u-(1-\frac{\delta_{l}^{2}}{\alpha^2})}{v-(1-\frac{\delta_{l}^{2}}{\alpha^2})}, \label{kernelCD2}\end{multline}
 where  $\mathcal{C}_{\mathrm{in}}$  encircles  $1-\delta_{1}^{2}/\alpha^2,\ldots, 1-\delta_{m}^{2}/\alpha^2$ and $4\mu/(1+\mu)^2$.  Next, we  prove Parts (i)--(iv) under the corresponding conditions  respectively.

 For Part (i) where $\mu N \ra \infty$ and  $0\leq \delta_{j}<\alpha$ for $j=1, \ldots, m$,  we choose   the contours such that  $ \mathcal{C}_{\mathrm{out}}$ goes around the origin    with  $\mathrm{Re}(z)<1$ for $z\in \mathcal{C}_{\mathrm{out}}$ and 
   $\mathcal{C}_{\mathrm{in}}$ is entirely  to the right side of  $\mathcal{C}_{\mathrm{out}}$.  
 In order to take limits smoothly,   we need to substitute   $K_{\nu}$ and $I_{\nu}$ into  \eqref{kernelCD2}   with their integral representations  respectively given by
     \eqref{Kfunction}  and \be I_{\kappa}(z) =\big(\frac{z}{2}\big)^{\kappa}\frac{1}{2\pi i}\int_{\mathcal{C}_{\mathrm{0}}}  ds  \,   s^{-\kappa-1}
 e^{s+\frac{z^2}{4s}}, \label{IBessel}\ee  
 which can be obtained by applying  the integral  representation of the reciprocal gamma function (cf. \eqref{inversegamma}) to  the RHS of \eqref{Ifunction}.
 Changing $u$ to $4\mu(1+\mu)^{-2}Nu$ and $v$ to $4\mu(1+\mu)^{-2}Nv$, we then use  Fubini's theorem to get
    \begin{align}
\frac{\mu }{N} K_N  \big( \frac{\mu }{N} \xi,  \frac{\mu }{N} \eta\big)&=\left(  \frac{\eta}{\xi}\right)^{\kappa/2}  \frac{1}{2\pi i}\int_{0}^{\infty}  dt  \int_{\mathcal{C}_{\mathrm{0}}}  ds  \,     t^{\kappa-1} s^{-\kappa-1}
 e^{s-t}  \widetilde{K}_{N}(\frac{\eta}{s},\frac{\xi}{t})  \label{kernelCD3}\end{align}
 where   \begin{align}
  \widetilde{K}_{N}(\frac{\eta}{s},\frac{\xi}{t})&=\exp\Big\{\frac{\eta }{s}\frac{(1+\mu)^2}{4\mu N}-\frac{\xi }{t}\frac{(1+\mu)^2}{4\mu N}\Big\}\,  \frac{1}{(2\pi i)^2}\int_{\mathcal{C}_{\mathrm{out}}} du \int_{\mathcal{C}_{\mathrm{in}}} dv\, e^{\frac{\xi u}{t}-\frac{\eta v}{s}}  \nonumber \\
 &\ \times \frac{1}{u-v}
\Big(\frac{v}{u}\Big)^{\nu+m} \Big(\frac{1-\frac{1}{Nu} }{1-\frac{1}{Nv}}\Big)^{N-m} 
  \prod_{l=1}^{m}\frac{u-\frac{(1+\mu)^2}{4\mu N}\big(1-\frac{\delta_{l}^{2}}{\alpha^2}\big)}{v-\frac{(1+\mu)^2}{4\mu N}\big(1-\frac{\delta_{l}^{2}}{\alpha^2}\big)}. \label{kernelCD4}\end{align}
    Here   $ \mathcal{C}_{\mathrm{out}}$  is  a counterclockwise contour around the origin and entirely to its right side    $\mathcal{C}_{\mathrm{in}}$  encircles  $1/N$ and $a_{l}:=(4\mu N)^{-1}(1+\mu)^{2}(1-\delta_{l}^{2}/\alpha^2)$ for $l=1, \ldots, m$. 
 
  To take  limit  as $N \ra \infty$ in \eqref{kernelCD4}, we need to deform the two contours. For this,  denote by $\mathcal{L}_{c_1,c_2;r}$  with $c_1<c_2$ and $r>0$ a  rectangular contour  connecting four points $(c_1, \pm r), (c_2, \pm r)$ in a counterclockwise direction.  Take $b_1, b_2$  such that $0<b_{1}<\min\{1/N, a_1, \ldots, a_m\}$ and $b_{2}>\max\{1/N, a_1, \ldots, a_m\}$, so we can specify $\mathcal{C}_{\mathrm{out}}$ and $ \mathcal{C}_{\mathrm{in}}$ with 
rectangular contours  $\mathcal{L}_{-b_{1}/2, b_{1}/2;2}$  and  $\mathcal{L}_{b_{1}, b_{2};1}$, respectively.   For convenience, let's  use  an abbreviated notation for the RHS of  \eqref{kernelCD4}.   We thus arrive at 
  \begin{align}
  \widetilde{K}_{N}&(\frac{\eta}{s},\frac{\xi}{t})=\int_{\mathcal{L}_{-b_{1}/2, b_{1}/2;2}} du \int_{\mathcal{L}_{b_{1}, b_{2};1}} dv\, \Big(\cdot\Big)  \nonumber \\
 &= \int_{\mathcal{L}_{-b_{1}/2, 2b_{2};2}} du \int_{\mathcal{L}_{b_{1}, b_{2};1}} dv\, \Big(\cdot\Big) -\int_{\mathcal{L}_{b_{1}/2, 2b_{2};2}} du \int_{\mathcal{L}_{b_{1}, b_{2};1}} dv\, \Big(\cdot\Big)  \label{kernelCD5-11}\\
 &= \int_{\mathcal{L}_{-b_{1}/2, 2b_{2};2}} du \int_{\mathcal{L}_{b_{1}, b_{2};1}} dv\, \Big(\cdot\Big),  \label{kernelCD5-12} 
 \end{align}
where  the second integral in  \eqref{kernelCD5-11} is actually equal to zero because the integrand has no pole with respect to $u$.  Moreover,  we can deform the resulting integral  \eqref{kernelCD5-12}  again and get 
 \begin{align}
  \widetilde{K}_{N}&(\frac{\eta}{s},\frac{\xi}{t})= \int_{\mathcal{L}_{-2b_{2}, 2b_{2};2}} du \int_{\mathcal{L}_{-b_{2}, b_{2};1}} dv\, \Big(\cdot\Big).  \label{kernelCD5-13} 
 \end{align}

 Since $\mu N \ra \infty$ as $N \ra \infty$, noting that $0<\mu\leq 1$ and $ 0\leq \delta_l<\alpha$ ($l=1, \ldots, m$), for $N$ large sufficient (for instance, $N>1$ and $\mu N>1$),  set $b_2=1$ in \eqref{kernelCD5-13},    
  application of Lebesgue's dominated convergence theorem provides us
    \begin{align}
 \widetilde{K}_{N}(\frac{\eta}{s},\frac{\xi}{t}) \rightarrow  \widetilde{K}_{\infty}(\frac{\eta}{s},\frac{\xi}{t}):=\frac{1}{(2\pi i)^2}  \int_{\mathcal{L}_{-2, 2;2}} du \int_{\mathcal{L}_{-1, 1;1}} dv\,  e^{\frac{\xi u}{t}-\frac{1}{u}-\frac{\eta v}{s}+\frac{1}{v}}\frac{1}{u-v} \Big(\frac{v}{u}\Big)^{\nu}
 . \label{4.5} \end{align}
  This limit has been identified as   the Bessel kernel   by  Desrosiers and Forrester (cf.  \cite[eqns (1.20) and (6.20)]{DF06}) with a specific   relation
 \be \widetilde{K}_{\infty}(s,t)= 4 \left(  \frac{\xi s}{\eta t}\right)^{\nu/2} K_{\mathrm{\nu}}^{(\mathrm{Bes})}\Big(\frac{4\eta}{s}, \frac{4\xi}{t}\Big), \label{4.6}\ee
from which the requested conclusion follows.

 In the case of  Part (ii) where  $\mu N \ra \tau/4$ with $\tau>0$ and
$ 1- \delta_{l}^{2}/\alpha^2  \ra \pi_l \in (0, 1)$ for $l=1, \ldots, m$,  for large $N$ sufficient  we see  all $1- \delta_{l}^{2}/\alpha^2  \in (0, r_1)$ with a  given positive number $r_1$ satisfying  $1>r_1>\max\{\pi_1, \ldots, \pi_m\}$.  Given  $1>r_2>r_1$,   let 
   $\mathcal{C}_{\mathrm{in}}$ and  $\mathcal{C}_{\mathrm{out}}$ be circles   with  radius $r_1$ and $r_2$ and center at the origin.
Note that  the involved function is continuous  in  the given bounded  contours   and as $N \ra \infty$
 \be \Big(   1-\frac{1}{z}\big(1-\frac{\delta^{2}}{\alpha^2}\big)\Big)^{N-m} \ra e^{-\frac{\tau }{z}},   \nonumber \ee
 take limit  in the integrand of the right-hand side of  \eqref{kernelCD2} and  we have  the required conclusion \eqref{IIlimit}.

 For Part (iii) where $\mu N \ra 0$ and $ 1- \delta_{l}^{2}/\alpha^2 =4\mu N\pi_l$ with   $\pi_l\in (0, \infty)$ for $l=1, \ldots, m$, change $u$ to $4\mu Nu$ and $v$ to $4\mu Nv$ in the integrand of the right-hand side of  \eqref{kernelCD2}, recalling  \eqref{finiterank} and  the assumptions on $\delta_{1}, \ldots, \delta_m$, we have
 \begin{align}
\frac{1 }{4N^2} K_N  \big( \frac{1 }{4N^2} \xi, & \frac{1 }{4N^2} \eta\big) =\frac{\mu}{N}\frac{2\alpha^2}{(2\pi i)^2}\int_{\mathcal{C}_{\mathrm{out}}} du \int_{\mathcal{C}_{\mathrm{in}}} dv\nonumber \\
&\,K_{-\kappa} \big(\frac{\alpha}{N} \sqrt{(1-4\mu N u)\xi}\big)
 \,  I_{\kappa}\big(\frac{\alpha}{N} \sqrt{(1-4\mu N v)\eta}\big) \frac{1}{u-v}\Big(\frac{v}{u}\Big)^{\nu+m} \nonumber\\ & \times
 \bigg(\frac{1-4\mu Nu}{1-4\mu Nv}\bigg)^{\kappa/2}  \bigg(\frac{1-\frac{1}{(1+\mu)^2Nu}}{1-\frac{1}{(1+\mu)^2Nv}}\bigg)^{N-m} \prod_{l=1}^{m}\frac{u-\pi_{l}}{v-\pi_l},  \nonumber\end{align}
 where  the two contours are chosen   such that $\mathrm{Re}(z)<(4\mu N)^{-1}$ for any $z$ in $ \mathcal{C}_{\mathrm{out}}$ and $\mathcal{C}_{\mathrm{in}}$, and  $\mathcal{C}_{\mathrm{in}}$ encircles  $0, (1+\mu)^{-2}N^{-1}, \pi_{1},\ldots, \pi_{m}$ and is wholly within $ \mathcal{C}_{\mathrm{out}}$ (here the same notations of the contours  are used, for simplicity).  Moreover, since $\mu N \ra 0$ as $N \ra \infty$,  for $N$ large enough we can always assume that   both the contours are selected and are independent of $N$. For instance, let $c_{0}=\max\{1, \pi_1, \ldots, \pi_m\}$, when $4\mu N\leq 1/(2c_{0}+8)$,  we can take $ \mathcal{C}_{\mathrm{out}}$ and $\mathcal{C}_{\mathrm{in}}$ as two circles centered at zero with radius $c_{0}+2$ and  $c_0+1$ respectively.
 
 We need to make use of asymptotic formulas of modified Bessel functions, see e.g.   \cite[10.40 (i)]{Olver10}. For any given constant $\delta$ such that  $0<\delta<\pi/2$, then as $ z\ra \infty$, with $\nu$ fixed,  the following hold uniformly with respect to  $\textrm{arg}(z)$ in the corresponding  sectors 
 \be I_{\nu}(z)= \frac{e^z}{\sqrt{2\pi z}} \Big(1+\mathcal{O}\big(\frac{1}{z}\big)\Big), \qquad 
 \  |\textrm{arg}(z)\leq \frac{1}{2}\pi -\delta,  \label{asymptoticsIK0} \ee
  and 
   \be K_{\nu}(z)=  \sqrt{\frac{\pi}{2z}} e^{-z} \Big(1+\mathcal{O}\big(\frac{1}{z}\big)\Big), \qquad  
   \  |\textrm{arg}(z)|\leq\frac{3}{2}\pi-\delta.  \label{asymptoticsIK}\ee

 Notice  the assumption   that  $\mu N \ra 0$ as $N \ra \infty$,  for any given two closed contours $\mathcal{C}_{\mathrm{out}}$ and $\mathcal{C}_{\mathrm{in}}$ (say, the two circles  with radius $c_{0}+2$ and  $c_0+1$ described previously), independent of $N$,  we can choose sufficiently large $N$ such that $|4\mu N u|\leq 1/2$ and $|4\mu N v|\leq 1/2$ uniformly for $u\in \mathcal{C}_{\mathrm{out}}$  and $v\in \mathcal{C}_{\mathrm{in}}$. These show that
  $|\textrm{arg}(\sqrt{1-4\mu N u})|\leq \pi/8$ and $|\textrm{arg}(\sqrt{1-4\mu N v})|\leq \pi/8$ for any $u\in \mathcal{C}_{\mathrm{out}}$, $v\in \mathcal{C}_{\mathrm{in}}$. Noting 
  $\alpha /N \ra \infty$ as $N\ra \infty$  and  the Taylor expansion
\be \frac{\alpha}{N} \sqrt{1-4\mu N v}=\frac{\alpha}{ N} -(1+\mu)v+\mathcal{O}(\mu N), \quad \mu N \ra 0, \nonumber\ee
 applying \eqref{asymptoticsIK0}  thus  gives rise to 

 \be I_{\kappa}\big(\frac{\alpha}{N} \sqrt{(1-4\mu N v)\eta}\big) \sim \frac{1}{\sqrt{2\pi}}\sqrt{\frac{N}{\alpha }} \eta^{-\frac{1}{4}}e^{\sqrt{\eta}\frac{\alpha}{N}  -\sqrt{\eta}v}
 \nonumber \ee
 uniformly for any $v\in \mathcal{C}_{\mathrm{in}}$.
 Likewise, we see from  \eqref{asymptoticsIK} that 
  \be K_{-\kappa}\big(\frac{\alpha}{N} \sqrt{(1-4\mu N u)\xi}\big) \sim  \sqrt{\frac{\pi}{2}}\sqrt{\frac{N}{\alpha }} \xi^{-\frac{1}{4}} e^{-\sqrt{\xi}\frac{\alpha}{N}  +\sqrt{\xi}u} \nonumber
 \ee
 uniformly  for any  $u\in \mathcal{C}_{\mathrm{out}}$.
 Taken together,  the desired  result immediately follows from application of Lebesgue's dominated convergence theorem.

 In the case of  Part (iv) where $\mu N \ra 0$ and $ 1- \delta_{l}^{2}/\alpha^2  \ra \pi_l\in (0,1)$ for $l=1, \ldots, m$,  
  as in Part (ii) for large $N$ sufficient  we can let 
   $\mathcal{C}_{\mathrm{in}}$  and  $\mathcal{C}_{\mathrm{out}}$   be  two circles    with  center and radius    independent of $N$. 
   Note that    $\mu N \ra 0$ implies
 \be \Big(   1-\frac{1}{z}\big(1-\frac{\delta^{2}}{\alpha^2}\big)\Big)^{N-m} \ra 1, \nonumber   \ee
 we get the required conclusion \eqref{IVlimit} by taking limit  in the integrand of the right-hand side of  \eqref{kernelCD2}.

Obviously,   the results hold uniformly for $\xi, \eta$ in  a given compact set  of $(0, \infty)$. Therefore, we have completed the proof of the given statement.
   \end{proof}

We now state a similar result associated with   correlation kernel    \eqref{kernelCDJ}.

\begin{theorem}  \label{hardlimitsJ}
 With  the kernel    \eqref{kernelCDJ} and  the assumptions given in \eqref{finiterank}, and with fixed nonnegative integers $\nu, \nu'$ and $\kappa$ such that $\nu+\nu'\geq \kappa$,  the following hold uniformly for any $\xi$ and $\eta$ in a compact set of $(0,\infty)$ as $N\rightarrow \infty$.

\begin{itemize}

 \item [(i)]  If $\mu N \ra \infty$, 
 then
 \be
\frac{1+\mu }{2N^2} K_N\Big( \frac{1+\mu }{2N^2} \xi, \frac{1+\mu }{2N^2} \eta\Big)\ra \Big(\frac{\eta}{\xi}\Big)^{\kappa/2}K_{\mathrm{I}}(\xi,\eta).  \nonumber \ee

 \item [(ii)]  If $\mu N \ra \tau/4$ with $\tau>0$ and
$ 1- \delta_{l}^{2}/\alpha^2  \ra \pi_l \in (0, 1)$ for $l=1, \ldots, m,$ then
 \be
 \frac{1}{\alpha N} K_{N}\Big(\frac{\xi}{\alpha N},  \frac{\eta}{\alpha N}\Big)\ra \Big(\frac{\eta}{\xi}\Big)^{\kappa/2} K_{\mathrm{II}}(\tau;\xi,\eta). \nonumber \ee

\item [(iii)]  If $\mu N \ra 0$ but $\mu N^{2} \ra \infty$, and
$ 1- \delta_{l}^{2}/\alpha^2=4\mu N  \pi_l$  with $\pi_l\in (0, \infty)$ for $l=1, \ldots, m,$ then
 \be
   \frac{ e^{\frac{1}{2\mu N} \sqrt{\xi}}}{e^{\frac{1}{2\mu N} \sqrt{\eta}}} \frac{1}{16\alpha \mu^{2}N^3}  K_N\Big(\frac{\xi}{16\alpha \mu^{2}N^3}, \frac{\eta}{16\alpha \mu^{2}N^3} \Big)\ra \Big(\frac{\eta}{\xi}\Big)^{\kappa/2}K_{\mathrm{III}}(\xi,\eta). \nonumber \ee

\item [(iv)]  If $\mu N \ra 0$ and
$ 1- \delta_{l}^{2}/\alpha^2  \ra \pi_l\in (0,1)$ for $l=1, \ldots, m,$ then for $m\geq 1$
 \be
 \frac{1}{\alpha N} K_{N}\Big(\frac{\xi}{\alpha N},  \frac{\eta}{\alpha N}\Big)\ra \Big(\frac{\eta}{\xi}\Big)^{\kappa/2} K_{\mathrm{IV}}(\xi,\eta). \nonumber\ee
\end{itemize}
 \end{theorem}

 \begin{proof} We proceed in a similar way as in the proof of Theorem \ref{hardlimits}.  First, recall the assumptions   \eqref{finiterank} and rewrite the kernel  \eqref{kernelCDJ} as
  \begin{multline}
K_N(x,y)=\frac{ \alpha }{(2\pi i)^2}\Big(\frac{y}{x}\Big)^{\kappa}\int_{\mathcal{C}_{\mathrm{out}}} du \int_{\mathcal{C}_{\mathrm{in}}} dv \, {f_2(\cdot, \kappa; \alpha (1-u)x)}\, {f_1(\cdot, \kappa; \alpha (1-v)y)}
   \\  \times        \frac{1}{u-v} \Big(\frac{v}{u}\Big)^{\nu+m} \bigg(\frac{1-\frac{1}{u}(1-\frac{\delta^{2}}{\alpha^2})}{1-\frac{1}{v}(1-\frac{\delta^{2}}{\alpha^2})}\bigg)^{N-m} \prod_{l=1}^{m}\frac{u-(1-\frac{\delta_{l}^{2}}{\alpha^2})}{v-(1-\frac{\delta_{l}^{2}}{\alpha^2})}. \label{kernelCD2J}\end{multline}

 For Part (i), without loss of generality we assume that $\mu<1$ and $ 1- \delta_{l}^{2}/\alpha^2  \ra \pi_l<1$ for $l=1, \ldots, m$ (otherwise, see the proof of Part (i) of Theorem \ref{hardlimits}),   then  we can  choose the contours such that
  $\mathcal{C}_{\mathrm{in}}$ is   wholly  within $\mathcal{C}_{\mathrm{out}}$.  In order to take limits smoothly,   substituting  \eqref{deff1-1}  and \eqref{deff2}, changing $t$ to $t/N$, $u$ to $(1-\delta^{2}/\alpha^{2})Nu$ and $v$ to $(1-\delta^{2}/\alpha^{2})Nv$, we then use  Fubini's theorem to get
    \begin{align}
\frac{1+\mu }{2N^2} K_N\Big( \frac{1+\mu }{2N^2} \xi, \frac{1+\mu }{2N^2} \eta\Big) &=\left(  \frac{\eta}{\xi}\right)^{\kappa}  \frac{1}{2\pi i}\int_{0}^{\infty}  \frac{dt}{t}  \int_{\mathcal{C}_{\mathrm{0}}}  \frac{ds}{s}  \,     t^{\kappa} s^{-\kappa}
 e^{s-t}  \widetilde{K}_{N}(s,t)  \label{kernelCD3J} \nonumber\end{align}
 where   
 \begin{align}
\widetilde{K}_{N}(s,t)&=N^{-\kappa-1}\frac{ \Gamma(\nu+\nu'+N+1)}{\Gamma(\nu+\nu'-\kappa+N)} \Big(1-\frac{t}{N}\Big)^{\nu+\nu'+N-\kappa-1}e^{t} 1_{[0,N]}(t) \int  du \int  dv  \nonumber\\
& \frac{1}{(2\pi i)^2}\,  \exp\Big\{\Big(u- \frac{(1+\mu)^2 }{4\mu N }\Big)\frac{\xi  }{t}\Big\}  \, \Big(1-\Big(\frac{(1+\mu)^2}{4\mu N} -v \Big)\frac{\eta v}{s}\Big)^{-\nu-\nu'-N-1} \nonumber\\
 &\ \times
     \frac{1}{u-v} \Big(\frac{v}{u}\Big)^{\nu+m} \Big(\frac{1-\frac{1}{Nu} }{1-\frac{1}{Nv}}\Big)^{N-m} \prod_{l=1}^{m}\frac{u-\frac{(1+\mu)^2}{4\mu N}\big(1-\frac{\delta_{l}^{2}}{\alpha^2}\big)}{v-\frac{(1+\mu)^2}{4\mu N}\big(1-\frac{\delta_{l}^{2}}{\alpha^2}\big)}. %%\label{kernelCD4J}
     \end{align}
     Since $\mu N \ra \infty$ as $N \ra \infty$, noting that $0<\mu\leq 1$ and $ 0\leq \delta_l<\alpha$ ($l=1, \ldots, m$),
  application of Lebesgue's dominated convergence theorem provides us
    \begin{align}
\widetilde{K}_{N}(s,t) \rightarrow  \widetilde{K}_{\infty}(s,t):=\frac{1}{(2\pi i)^2}\int_{\mathcal{C}_{\mathrm{out}}} du \int_{\mathcal{C}_{\mathrm{in}}} dv\, e^{\frac{\xi u}{t}-\frac{1}{u}-\frac{\eta v}{s}+\frac{1}{v}}\frac{1}{u-v} \Big(\frac{v}{u}\Big)^{\nu}, \nonumber \end{align}
  from which the requested conclusion follows (cf. eqn\eqref{4.6}).

 We easily verify    Parts (ii) and (iv) as in the proof of Theorem \ref{hardlimits}.  For Part (iii), with \eqref{deff1-1}  and \eqref{deff2}  in mind,   change of variables  $t \mapsto t/(\mu N^2), t \mapsto t/(\mu N^2)$, $u \mapsto4\mu Nu$ and $v \mapsto 4\mu Nv$ in   \eqref{kernelCD2J}, as $N \rightarrow \infty$ we have
 \begin{align}
 \frac{1}{16\alpha \mu^{2}N^3}  &K_N\Big(\frac{\xi}{16\alpha \mu^{2}N^3}, \frac{\eta}{16\alpha \mu^{2}N^3} \Big)\nonumber \\
&\,  \sim \left(  \frac{\eta}{\xi}\right)^{\kappa} \frac{1}{4\mu N}\frac{1}{\Gamma(\kappa+1)} \frac{1}{(2\pi i)^2}\int_{\mathcal{C}_{\mathrm{out}}} du \int_{\mathcal{C}_{\mathrm{in}}} dv \, g_{1,N}(u) g_{2,N}(v)    \nonumber\\ & \times
    \bigg(\frac{1-\frac{1}{(1+\mu)^2Nu}}{1-\frac{1}{(1+\mu)^2Nv}}\bigg)^{N-m} \frac{1}{u-v}\Big(\frac{v}{u}\Big)^{\nu+m}\prod_{l=1}^{m}\frac{u-\pi_{l}}{v-\pi_l}, \label{comb1} \end{align}
 where
 \be g_{1,N}(u)=\int_{0}^{\mu N^2} dt \, t^{\kappa-1}\Big(1-\frac{t}{\mu N^2}\Big)^{\nu+\nu'+N-\kappa-1} \exp\Big\{ \frac{u \xi  }{4t}- \frac{\xi }{16\mu N t } \Big\}   \nonumber\ee
 and
  \be g_{2,N}(v)=\frac{1}{2\pi i}\int_{\mathcal{C}_{0}} ds\,  s^{-\kappa-1}  e^{\frac{s}{\mu N}} \Big(1+\frac{v\eta}{4Ns}- \frac{\eta }{16\mu N^{2} s } \Big)^{-\nu-\nu'-N-1}.  \nonumber \ee

 It suffices to find the leading coefficients for both functions. For this purpose  we  use the method of steepest decent; see e.g.   \cite{wong01}.  For  $g_{1,N}(u)$, we use the inequality  $1-x\leq e^{x}$ ($0\leq x\leq 1$) to get
 \begin{align}|g_{1,N}(u)| \leq \int_{0}^{\mu N^2} dt \, |e^{\frac{u \xi  }{4t}}| t^{\kappa-1}  \exp\Big\{ -\frac{1}{\mu N}\Big(t+ \frac{\xi }{16 t }\Big)-(\nu+\nu'-\kappa-1)\frac{t}{\mu N^2}\Big\}. \nonumber\end{align}
 Note that the function $t+  \xi/(16 t)$ attains its unique minimum at $t_0=\sqrt{\xi}/4$ over $(0,\infty)$ and both $1/(\mu N)$ and $\mu N^2$ go  to infinity,  the leading contribution must come from the neighbourhood of  $t_0$. By   Taylor expansion, we easily see that
 \be g_{1,N}(u)\sim 2^{1-2\kappa}\sqrt{\pi \mu N}  \xi^{\frac{\kappa}{2}-\frac{1}{4}} e^{ \sqrt{\xi}u-\frac{\sqrt{\xi}}{2\mu N}}.  \label{comb2}\ee

  For  $g_{2,N}(v)$, noting
  \begin{align} g_{2,N}(v) \sim
  \frac{1}{2\pi i}\int_{\mathcal{C}_{0}}  ds\,  s^{-\kappa-1}   \exp\Big\{
   \frac{1}{\mu N} \big(s+\frac{\eta}{16s}\big)-\frac{v\eta}{4s}
   \Big\}, \nonumber
  \end{align}
  let  $\mathcal{C}_{0}$ be a circle of radius $\sqrt{\eta}/4$, it is easy to verify that $\mathrm{Re}\{s+\frac{\eta}{16s}\}$ attains a unique maximum at  $s_0=\sqrt{\eta}/4$. Thus the  steepest decent   argument leads us to
  \be g_{2,N}(v)\sim 2^{2\kappa}\sqrt{ \mu N/\pi}  \xi^{-\frac{\kappa}{2}-\frac{1}{4}} e^{ -\sqrt{\eta}v+\frac{\sqrt{\xi}}{2\mu N}}.  \label{comb3}\ee

   Substitution of  \eqref{comb2} and  \eqref{comb3} in \eqref{comb1} completes Part (iii).

Obviously,   the results hold uniformly for $\xi, \eta$ in  a given compact set  of $(0, \infty)$.
   \end{proof}

Compare Part (iii) in  Theorems \ref{hardlimits} and  \ref{hardlimitsJ},   there is a technical restriction on the rate of  $\mu N$ in the latter. We believe this can be removed such that the same result holds true as in the former.

\section{On the four limiting  kernels} \label{4kernels}

\subsection{Comparison}
We first introduce a few families of contour integrals and rewrite the kernels defined as before.  Setting
 \be \tilde{\Lambda}^{(k)}_{\mathrm{II}}(x)=\frac{1}{\pi i}\int_{\mathcal{C}_{0}} du \,K_{-\kappa}(2\sqrt{(1-u)x}) (1-u)^{\kappa/2}u^{-\nu-m} e^{-\frac{\tau}{u}}  \prod_{l=1}^{k-1}(u- \pi_{l}),\ee
 \be \Lambda^{(k)}_{\mathrm{II}}(x)=\frac{1}{2\pi i}\int_{\mathcal{C}_{\pi}} dv \, I_{\kappa}(2\sqrt{(1-v)x}) (1-v)^{-\kappa/2} v^{\nu+m} e^{\frac{\tau}{v}}  \prod_{l=1}^{k}\frac{1}{v- \pi_{l}}\ee
where $\mathcal{C}_{0}$ denotes a contour enclosing the origin and   $\mathcal{C}_{\pi}$  encloses  $\pi_1, \ldots, \pi_k$,
and
\begin{align}
K_{\mathrm{II}}^{(0)}(\tau;\xi,\eta)&= \frac{2}{(2\pi i)^2}\int_{\mathcal{C}_{\mathrm{out}}} du \int_{\mathcal{C}_{\mathrm{in}}} dv \,  K_{-\kappa}(2\sqrt{(1-u)\xi})\,  I_{\kappa}(2\sqrt{(1-v)\eta}) \nonumber\\  &  \quad \times e^{-\frac{\tau}{u}+\frac{\tau}{v}}\frac{1}{u-v} \Big(\frac{1-u}{1-v}\Big)^{\kappa/2}\Big(\frac{u}{v}\Big)^{-\nu-m}, \label{kernelcrit0}\end{align}
then the use of the identity (see e.g.   \cite[Eq.(5.12)]{DF06})
  \be \frac{1}{u-v}\prod_{l=1}^{m}\frac{u-\pi_l}{v-\pi_l}=\frac{1}{u-v}+\sum_{k=1}^m \frac{\prod_{l=1}^{k-1}(u-\pi_l)}{\prod_{l=1}^{k}(v-\pi_l)} \label{identicalrelation}\ee
  immediately gives us
\be K_{\mathrm{II}}(\tau;\xi,\eta)=K_{\mathrm{II}}^{(0)}(\tau;\xi,\eta)+\sum_{k=1}^{m}\tilde{\Lambda}^{(k)}_{\mathrm{II}}(\xi)\Lambda^{(k)}_{\mathrm{II}}(\eta). \label{IIsum}\ee
Likewise, setting
 \be \tilde{\Lambda}^{(k)}_{\mathrm{III}}(x)= \frac{1}{  2\pi i } \frac{1}{2\xi^{\frac{1}{4}}} \int_{\mathcal{C}_{0}} du \, e^{\sqrt{\xi}u-\frac{1}{u}} u^{-\nu-m} \prod_{l=1}^{k-1}(u- \pi_{l}),\ee
 \be \Lambda^{(k)}_{\mathrm{III}}(x)= \frac{1}{  2\pi i } \frac{1}{2\eta^{\frac{1}{4}}} \int_{\mathcal{C}_{\pi}} dv \, e^{-\sqrt{\xi}v+\frac{1}{v}} v^{\nu+m} \prod_{l=1}^{k}\frac{1}{v- \pi_{l}}\ee
and
\begin{align}
K_{\mathrm{III}}^{(0)}(\xi,\eta) &= \frac{2}{ (2\pi i)^2} \frac{1}{4(\xi \eta)^{\frac{1}{4}}} \int_{\mathcal{C}_{\mathrm{out}}} du \int_{\mathcal{C}_{\mathrm{in}}} dv \, e^{\sqrt{\xi}u-\sqrt{\eta}v-\frac{1}{u}+\frac{1}{v}}   \nonumber \\
 & \quad \times \frac{1}{u-v}\Big(\frac{u}{v}\Big)^{-\nu-m}, \label{kernelsup0}\end{align}
then
\be K_{\mathrm{III}}(\xi,\eta)=K_{\mathrm{III}}^{(0)}(\xi,\eta)+2\sum_{k=1}^{m}\tilde{\Lambda}^{(k)}_{\mathrm{III}}(\xi)\Lambda^{(k)}_{\mathrm{III}}(\eta). \label{IIIsum}\ee
Again, by defining
 \be \tilde{\Lambda}^{(k)}_{\mathrm{IV}}(x)=\frac{1}{\pi i}\int_{\mathcal{C}_{0}} du \,K_{-\kappa}(2\sqrt{(1-u)x})  (1-u)^{\kappa/2}  u^{-\nu-m}    \prod_{l=1}^{k-1}(u- \pi_{l}),\ee
 \be \Lambda^{(k)}_{\mathrm{IV}}(x)=\frac{1}{2\pi i}\int_{\mathcal{C}_{\pi}} dv \, I_{\kappa}(2\sqrt{(1-v)x})  (1-v)^{-\kappa/2} v^{\nu+m}   \prod_{l=1}^{k}\frac{1}{v- \pi_{l}}\ee
we have
\be K_{\mathrm{IV}}(\xi,\eta)=\sum_{k=1}^{m}\tilde{\Lambda}^{(k)}_{\mathrm{IV}}(\xi)\Lambda^{(k)}_{\mathrm{IV}}(\eta). \label{IVsum}\ee

 Next, we compare the four  limiting kernels   with the known limiting kernels in random matrix theory one after another.
When $\kappa=0$ and $m=0$, according to  the result of  Akemann and Strahov (cf. \cite[Theorem 3.9]{AS15}), the kernel $K_{\mathrm{I}}(x,y)$ is expected to be the Meijer G-kernel $K_{\nu,0}(x,y)$ up to a   transformation  like $f(x)/f(y)$ for some function $f(x)$ where
\begin{align}K_{\nu,\kappa}(x,y)  =&
 {1 \over (2 \pi i)^2} \int_{-1/2 - i \infty}^{-1/2 + i \infty} du  \oint_{\Sigma} dv \, {\sin \pi u \over \sin \pi v} \nonumber \\
 & \ \times     { \Gamma(u+1 ) \Gamma(\nu + u+1 )\Gamma(\kappa + u+1 ) \over   \Gamma(v+1) \Gamma(\nu + v+1 )\Gamma(\kappa + v+1 ) }
     {x^v y^{-u-1} \over u - v} \label{MeijerG2}
\end{align}
with $\Sigma$ a contour enclosing the positive real axis but not $u$. Actually, since the case of $\mu=1$ (cf.  \eqref{matrixpdf} and  \eqref{finiterank} in Sect. \ref{sectionintroduction}) reduces to the product of two independent Gaussian  rectangular matrices,  $K_{\mathrm{I}}(x,y)$   is strongly believed to be $K_{\nu,\kappa}(x,y)$, which was first found by Kuijlaars and Zhang \cite{KZ} in this context,   up to a   factor  $f(x)/f(y)$. They are indeed equal  according to the following proposition. Actually,  this  type of convolution representation has been obtained in  the product of two independent random matrices   for finite   matrix size $N$,  see \cite[Theorem 2.8(b)]{CKW15}.  Thus the limiting case is also expected. 
\begin{prop}  For the correlation kernels    \eqref{kernelsub} and \eqref{MeijerG2}, we have
\be
K_{\mathrm{I}}(\xi,\eta) = \left(   \eta/\xi \right)^{\kappa/2}   K_{\nu,\kappa}(\eta,\xi).  \label{Ikernelrelation}
\ee
  \end{prop}
\begin{proof} Start from the representation of the Bessel kernel (see e.g. \cite[Example 3.1]{Bo98} and \cite[Sect. 5.3]{KZ})
\be 4 K_{\mathrm{\nu}}^{(\mathrm{Bes})}\big(4x,4y\big) =\int_{0}^{1}  J_{\nu}(2\sqrt{xw})J_{\nu}(2\sqrt{yw}) dw,\ee
we have
\begin{align}
\left(   \eta/\xi \right)^{-\kappa/2} K_{\mathrm{I}}(\xi,\eta) &= \int_{0}^{1} dw  \frac{1}{2\pi i}  \int_{\mathcal{C}_{\mathrm{0}}}  ds  \,     s^{-\kappa-1}
 e^{s} \left(  \frac{\eta w }{s}\right)^{-\nu/2} J_{\nu}\Big(2\sqrt{\frac{\eta w}{s}}\Big) \nonumber \\
  & \quad \times       \int_{0}^{\infty}  dt\, t^{\kappa-1}  e^{-t}  \left(  \frac{\xi w}{t}\right)^{\nu/2} J_{\nu}\Big(2\sqrt{\frac{\xi w}{t}}\Big).
 \end{align}

Integrate term by term and then use the relation between hypergeometric functions and Meijer G-functions (cf. \cite[Sect. 5.2]{Luke}), we get
\begin{align}
 &\frac{1}{2\pi i}  \int_{\mathcal{C}_{\mathrm{0}}}  ds  \,     s^{-\kappa-1}
 e^{s} \left(  \frac{\eta w }{s}\right)^{-\nu/2} J_{\nu}\Big(2\sqrt{\frac{\eta w}{s}}\Big) \nonumber\\
 &=\frac{1}{\Gamma(\kappa+1)} \frac{1}{\Gamma(\nu+1)} {_{0}F_{2}}(\kappa+1,\nu+1;-\eta w) \nonumber \\
 &=G^{1,0}_{0,3} \Big ({\underline{\hspace{0.5cm}}
 \atop 0, -\nu,-\kappa} \Big |\eta w \Big ).
 %%%\frac{1}{\Gamma(\nu+1)}\frac{1}{2\pi i}  \int_{\mathcal{C}_{\mathrm{0}}}  ds  \,     s^{-\kappa-1}  e^{s}   {_{0}F_{1}}\Big(2\sqrt{\frac{\eta w}{s}}\Big)
  \end{align}
 On the other hand, noting
 \be \left(  \frac{\xi w}{t}\right)^{\nu/2} J_{\nu}\Big(2\sqrt{\frac{\xi w}{t}}\Big)=G^{1,0}_{0,2} \Big ({\underline{\hspace{0.5cm}}
 \atop \nu,0} \Big |\frac{\xi w}{t} \Big ), \ee
  the Mellin convolution formula (see e.g. \cite[Appendix eqn(A.3)]{KS14}) gives us
 \be \int_{0}^{\infty}  dt\, t^{\kappa-1}  e^{-t} G^{1,0}_{0,2} \Big ({\underline{\hspace{0.5cm}}
 \atop \nu,0} \Big |\frac{\xi w}{t} \Big )=G^{2,0}_{0,3} \Big ({\underline{\hspace{0.5cm}}
 \atop \nu, \kappa, 0} \Big | \xi w  \Big ).\ee

Therefore,
\begin{align}
\left(   \eta/\xi \right)^{-\kappa/2} K_{\mathrm{I}}(\xi,\eta) &= \int_{0}^{1} dw \,G^{1,0}_{0,3} \Big ({\underline{\hspace{0.5cm}}
 \atop 0, -\nu,-\kappa} \Big |\eta w \Big ) G^{2,0}_{0,3} \Big ({\underline{\hspace{0.5cm}}
 \atop \nu, \kappa, 0} \Big | \xi w  \Big ). \label{IMrep}
 \end{align}
 By  Theorem 5.3 of \cite{KZ} (noting $\nu_0=0$ therein), the RHS of \eqref{IMrep} is indeed  another integral  representation of the kernel $K_{\nu,\kappa}(\eta,\xi)$, from which     the desired result  immediately  follows.
  \end{proof}

 Here it's worth stressing that the   Meijer G-kernels $K_{\nu,\kappa}(x,y)$ already appeared in the works of  Bertola,    Gekhtman  and   Szmigielski on the Cauchy-Laguerre two matrix model \cite{BGS14}, Kuijlaars and Zhang \cite{Ku15} on products of two independent Gaussian  rectangular matrices, Forrester on the product with the inverse \cite{Fo14}. This shows that the kernel  is  universal.  It's probably  worth  pointing out  that although the  Borodin's kernel from  \cite{Bo98}  can be written in terms of Meijer G-functions, it does not agree with kernels stemming from products of random matrices as the indices obtained are different. See   \cite{KS14} for the inter-relation between Borodin's kernel and Meijer G-kernels.

Under the same conditions of $\kappa=0$ and $m=0$, at the critical scale of $\mu=g/N$ with $g\in (0, \infty)$,  for the rescaled kernel  $(1/(4N^2))K_{N}(x^2/(4N^2), y^2/(4N^2))$ Akemann and Strahov  obtained the hard edge limiting kernel defined by
\begin{align}\mathbb{S}(x,y;g)&= {4 \over (x^{2}-y^{2})g}
 {1 \over (2 \pi i)^2}  \oint_{\Sigma} du  \oint_{\Sigma} dv \,  {\Gamma(-u)\Gamma(-v)   \over    \Gamma(u+\nu +1 )\Gamma(v+\nu +1 )} x^v y^{u+\nu} \nonumber \\
 & \quad  \times  \Big ( A(u,v,\nu)-g(u^{2}+v^{2}-uv+\nu u)\Big)I_{v}\Big(\frac{x}{2g}\Big) K_{u+\nu}\Big(\frac{y}{2g}\Big) \label{ASrep}
\end{align}
 where  $\Sigma$ is a contour enclosing the positive real axis and
 \be  A(u,v,\nu)=\frac{1}{4}(v-u)\big(u^{2}+v^{2}+(\nu-1)(u+v)-\nu\big),\ee
 see \cite[Theorem 1.5 (b)]{AS15b}. In this case it  remains as a challenge for us  to directly verify   the equivalence of both the critical kernels $\mathbb{S}(x,y;g)$ and $K_{\mathrm{II}}^{(0)}(\tau;x,y)$. However, using integral representations of Bessel functions and noting \eqref{4.5} and \eqref{4.6}, it is easy to rewrite  the   kernel defined by \eqref{kernelcrit0} in terms of the Bessel kernel  as
  \begin{align}
K_{\mathrm{II}}^{(0)}(\tau;\xi,\eta) &= \left(  \frac{\xi}{\eta}\right)^{\kappa/2}  \frac{1}{2\pi i}\int_{0}^{\infty}  dt  \int_{\mathcal{C}_{\mathrm{0}}}  ds  \,     t^{\kappa-1} s^{-\kappa-1}
 e^{\eta s-\xi t+\frac{1}{s}-\frac{1}{t}}  \nonumber \\ & \quad \times    4\tau \left(  \frac{ s}{ t}\right)^{(\nu+m)/2} K_{\mathrm{\nu+m}}^{(\mathrm{Bes})}\Big(\frac{4\tau}{s}, \frac{4\tau}{t}\Big). \label{kernelsub-2}\end{align}

  With  change of variables, the kernel $K_{\mathrm{III}}$   has been identified by Desrosiers and Forrester \cite{DF06} as the hard edge limiting kernel  for the spiked complex sample covariance matrices. In particular,  we have the following relation (cf.  \cite[eqns (1.20) and (6.20)]{DF06})
 \be K_{\mathrm{III}}^{(0)}(\xi,\eta)=( \xi/\eta)^{(\nu+m)/4} 2(\xi \eta)^{-1/4}  K_{\mathrm{\nu+m}}^{(\mathrm{Bes})}\big(4\sqrt{\eta},4\sqrt{\xi}\big).
 \ee

The fourth kernel $K_{\mathrm{IV}}(x,y)$  is essentially the kernel \eqref{kernelCD} for the product of two coupled Gaussian random matrices but with $N\mapsto m, \alpha \mapsto 1$ and $\delta_{l}^{2} \mapsto 1-\pi_l,  l=1, \ldots, m$. Moreover,  as the correlation kernel of a determinantal point process it  corresponds to the joint PDF given in \eqref{eigenvaluepdfGinibre}. So in that sense,  it appears as one of  limiting kernels for the smallest  singular values in random matrix theory,    like  the   Gaussian Unitary Ensemble with source   for the  largest eigenvalues  or the noncentral Wishart matrices (also being called as shifted mean chiral Gaussian matrices)  for the smallest  singular values; see \cite{BBP} and \cite{FL15b}. Particularly for the case of $m=1$, note  \eqref{eigenvaluepdfGinibre} and Remark \ref{normconvergence}, under the same assumptions as in Theorem \ref{hardlimits} (iv) we have  
\be \mathbb{P}( x_1 \leq 4 \mu^{2}y,    \ldots,   x_N \leq 4 \mu^{2} y) \rightarrow  \frac{2(1-\pi_{1})}{\Gamma(\nu+1)\pi_{1}^{\kappa/2}}\int_{0}^{y} t^{\frac{\nu}{2}} K_{\nu-\kappa}(2\sqrt{t}) I_{\kappa}(2 \sqrt{ \pi_1 t})dt \label{goodform}\ee
as $N\rightarrow \infty$.

  Finally, we conclude this subsection with a transition from the critical kernel   $K_{\mathrm{II}}(x,y)$ to the other three kernels, which shows that $K_{\mathrm{II}}(x,y)$ is an interpolation between them.
This is to be expected, as then the parameter   effectively $\mu\sim \tau/(4N)$ and the coupled product  tends to  the classical Laguerre Unitary Ensemble  as $\mu \to 0$ while  it corresponds to   the product of two independent Gaussian random matrices as $\mu \to 1$. For $\kappa=0$ and $m=0$, similar resutls been obtained by   Akemann and Strahov \cite{AS15b}.

\begin{theorem}   \label{hardlimitstransition}
 With  the   kernels defined in \eqref{kernelsub}--\eqref{kernelsupsup}, the following hold  uniformly for any $x$ and $y$ in a compact set of $(0,\infty)$.

\begin{itemize}

 \item [(i)]  \be \lim_{\tau\ra \infty}\frac{1}{\tau}K_{\mathrm{II}}(\tau;\frac{x}{\tau},\frac{y}{\tau}) =K_{\mathrm{I}}(x,y). \nonumber \ee
 \item [(ii)] Given $q\leq m$, suppose that   $\pi_{l}=\tau \hat{\pi}_l $ for $l=1, \ldots, q$ and  $\pi_{q+1}, \ldots, \pi_{m}$  are fixed,  then
 \be \lim_{\tau\ra 0} e^{\frac{2}{\tau}(\sqrt{x}-\sqrt{y})}  \frac{1}{\tau^2}K_{\mathrm{II}}(\tau;\frac{x}{\tau^2},\frac{y}{\tau^2}) =K_{\mathrm{III}}(x,y)|_{m\mapsto q, \nu\mapsto \nu+m-q, \pi\mapsto \hat{\pi}}. \nonumber \ee
 \item [(iii)] \be \lim_{\tau\ra 0}K_{\mathrm{II}}(\tau;x,y) =K_{\mathrm{IV}}(x,y).\nonumber \ee
\end{itemize}
 \end{theorem}

\begin{proof}
For Part (i),  change $u$ to $\tau u$ and $v$ to $\tau v$ in the integrand of \eqref{kernelcrit},
we have
\begin{align}
\frac{1}{\tau}K_{\mathrm{II}}(\tau;\frac{x}{\tau},\frac{y}{\tau})&=\frac{2}{(2\pi i)^2}\int_{\mathcal{C}_{\mathrm{out}}} du \int_{\mathcal{C}_{\mathrm{in}}} dv \,K_{-\kappa} \Big(2 \sqrt{\big(1-\tau u\big)x/\tau}\Big) \,  I_{\kappa}\Big(2 \sqrt{\big(1-\tau v\big)y/\tau}\Big)\nonumber \\
& \, \times e^{-\frac{1}{u}+\frac{1}{v}} \frac{1}{u-v}\Big(\frac{v}{u}\Big)^{\nu+m} \big(\frac{1-\tau u}{1-\tau v}\big)^{\kappa/2}  \prod_{j=1}^{m}\frac{\tau u-\pi_j}{\tau v-\pi_j },  \label{transitionI}  \end{align}
where $\mathcal{C}_{\mathrm{out}}$   is a  simple  counterclockwise  contour around  the origin with  $\mathrm{Re}(z)<1/\tau,  \forall z\in \mathcal{C}_{\mathrm{out}}$ 
and entirely within it  $\mathcal{C}_{\mathrm{in}}$ is a counterclockwise contour encircling $0, \pi_{1}/\tau,\ldots,  \pi_{m}/\tau$.
Next,  we will use the similar argument  as in the proof of Part (i) of Theorem \ref{hardlimits}  to complete  it. 

Substitute   $K_{\nu}$ and $I_{\nu}$ into  \eqref{transitionI}   with their integral representations  respectively given by
     \eqref{Kfunction}  and  \eqref{IBessel}. 
 Use  Fubini's theorem  and we   rewrite the integral  appearing in \eqref{transitionI}   as 
    \begin{align}
\frac{1}{\tau}K_{\mathrm{II}}(\tau;\frac{x}{\tau},\frac{y}{\tau})&=\left(  \frac{y}{x}\right)^{\kappa/2}  \frac{1}{2\pi i}\int_{0}^{\infty}  dt  \int_{\mathcal{C}_{\mathrm{0}}}  ds  \,     t^{\kappa-1} s^{-\kappa-1}
 e^{s-t}  \widetilde{K}(\tau;\frac{y}{s},\frac{x}{t})  \label{transionI1}
  \end{align}
 where   \begin{align}
 \widetilde{K}(\tau;\frac{y}{s},\frac{x}{t})&= \frac{1}{(2\pi i)^2}\int_{\mathcal{C}_{\mathrm{out}}} du \int_{\mathcal{C}_{\mathrm{in}}} dv\,  e^{\frac{y }{\tau s}-\frac{x }{\tau t}}\,  e^{\frac{x u}{t}-\frac{yv}{s}-\frac{1}{u}+\frac{1}{v}}  \nonumber \\
 &\quad  \times \frac{1}{u-v}
\Big(\frac{v}{u}\Big)^{\nu+m}
  \prod_{l=1}^{m}\frac{u-\pi_{l}/\tau}{v-\pi_{l}/\tau}. \label{transitionI2}\end{align}
   Note that it is unnecessary to assume  $\mathrm{Re}(z)<1/\tau$  for $ z\in \mathcal{C}_{\mathrm{out}}$ in  \eqref{transitionI2}. In particular,  it is seen from $ 0< \pi_l<1$ ($l=1, \ldots, m$) that when $\tau>2$  we can choose    $\mathcal{C}_{\mathrm{in}}$ and  $\mathcal{C}_{\mathrm{out}}$ as two circles   with  radius $1$ and $2$ and both with center at the origin.
Note that  the involved function is continuous  in  the given contours   and as $\tau  \ra \infty$ 
  application of Lebesgue's dominated convergence theorem provides us
    \begin{align}
\widetilde{K}(\tau;\frac{y}{s},\frac{x}{t}) \rightarrow  \frac{1}{(2\pi i)^2}  \int_{\mathcal{C}_{\mathrm{out}}} du \int_{\mathcal{C}_{\mathrm{in}}} dv\,  e^{\frac{x u}{t}-\frac{1}{u}-\frac{y v}{s}+\frac{1}{v}}\frac{1}{u-v} \Big(\frac{v}{u}\Big)^{\nu}. \nonumber \end{align}
  This limit has been identified as   the Bessel kernel   by  Desrosiers and Forrester \cite{DF06} and  the requested conclusion then follows (cf.  \eqref{4.5} and \eqref{4.6} in Sect. \ref{sectionhardlimit}).

For Part (ii),   change $u$ to $\tau u$ and $v$ to $\tau v$ in the integrand of \eqref{kernelcrit},
we have
\begin{align}
\frac{1}{\tau^2}K_{\mathrm{II}}(\tau;\frac{x}{\tau^2},\frac{y}{\tau^2})&=\frac{2}{(2\pi i)^2 \tau}\int_{\mathcal{C}_{\mathrm{out}}} du \int_{\mathcal{C}_{\mathrm{in}}} dv \,K_{-\kappa} \Big(\tfrac{2}{\tau} \sqrt{\big(1-\tau u\big)x}\Big) \,  I_{\kappa}\Big(\tfrac{2}{\tau} \sqrt{\big(1-\tau v\big)y}\Big)\nonumber \\
& \,  \times e^{-\frac{1}{u}+\frac{1}{v}} \frac{1}{u-v}\Big(\frac{v}{u}\Big)^{\nu+m} \big(\frac{1-\tau u}{1-\tau v}\big)^{\kappa/2}  \prod_{l=1}^{q}\frac{u-\hat{\pi}_l}{v-\hat{\pi}_l } \prod_{j=q+1}^{m}\frac{\tau u-\pi_j}{\tau v-\pi_j }.   \nonumber \end{align}
As $\tau \ra 0$, by \eqref{asymptoticsIK0}  and  \eqref{asymptoticsIK}  simple calculation  gives  us
\be \,K_{-\kappa} \Big(\tfrac{2}{\tau} \sqrt{\big(1-\tau u\big)x}\Big) \,  I_{\kappa}\Big(\tfrac{2}{\tau} \sqrt{\big(1-\tau v\big)y}\Big) \sim \frac{\tau}{4(xy)^{\frac{1}{4}}}e^{-\frac{2}{\tau}\sqrt{x}+\sqrt{x} u +\frac{2}{\tau}\sqrt{y}-\sqrt{y} v}. \nonumber \ee
 Obviously, the  function of variables $u$ and $v$ in the limit above  is continuous in the bounded contours $ \mathcal{C}_{\mathrm{out}}$ and $\mathcal{C}_{\mathrm{in}}$,    application of Lebesgue's dominated convergence theorem thus provides us  Part (ii) as $\tau\ra 0$.

Lastly, taking limit in the definition of  \eqref{kernelcrit},  we have Part (iii).
\end{proof}

\subsection{Integrable form of the critical kernel}
Recall that a correlation kernel $K(x, y)$ is called integrable  in the sense of Its, Isergin,
Korepin and Slavnov \cite{IIKS90}  if it can be represented as
\be  K(x, y) = \frac{\sum_{i=1}^{k} f_{i}(x)g_{i}(y)}{x-y}, \quad \mathrm{with} \quad   \sum_{i=1}^{k} f_{i}(x)g_{i}(x)=0  \ee
for some integer $k\geq 2$ and certain functions $f_i$ and $g_i$. The kernels of standard universality  classes in Random Matrix Theory, for instance, sine, Airy and Bessel kernels, all belong to the class of integrable kernels.   Recently, the Meijer G-kernel $K_{\mathrm{I}}$ (cf. Eqns \eqref{MeijerG2} and \eqref{Ikernelrelation})  has turned out  to be integrable, see \cite{BGS14,KZ} or \cite{St14,WF16} for relevant Hamiltonian differential equations.  On the other hand, noting the fact that
\be \sum_{i=1}^{k} f_{i}(x)g_{i}(y)=\frac{1}{x-y}\sum_{i=1}^{k}\big( xf_{i}(x)g_{i}(y)-f_{i}(x)y g_{i}(y)\big)\ee
and $K^{(0)}_{\mathrm{III}}$ is integrable,  it is  easy to verify from \eqref{IIIsum} and \eqref{IVsum} that $K_{\mathrm{III}}$ and $K_{\mathrm{IV}}$ are also integrable.
 As for the critical kernel $K_{\mathrm{II}}$, we argue that
the new limiting kernel $K^{(0)}_{\mathrm{II}}$ can be represented in an integrable form in terms of two functions and their derivatives up to third order, and so does the kernel $K_{\mathrm{II}}$ because of  \eqref{IIsum}.  However, in the case of $\kappa=\nu=0$ Akemann and Strahov gave an integrable form  of their limiting kernel \eqref{ASrep} (cf.  \cite[Sect. 1.5]{AS15b}).

In order to state the integrable form of the critical kernel, we need two functions defined by integrals involving Bessel functions for non-negative integers $\alpha, \kappa$
\be f(x)=\int_{0}^{\infty}  dt    \,     t^{-(\frac{\alpha}{2}-\kappa)-3}
 e^{-x t-\frac{1}{t}}      J_{\mathrm{\alpha}}\Big( \sqrt{\frac{4\tau}{t}}\Big), \qquad x>0,  
  \label{ffunction}\ee
and
 \be g(x)=\frac{1}{2\pi i}   \int_{\mathcal{C}_{\mathrm{0}}}  ds  \,   s^{(\frac{\alpha}{2}-\kappa)-3}
 e^{x s+\frac{1}{s}}      J_{\mathrm{\alpha}}\Big( \sqrt{\frac{4\tau}{s}}\Big), \qquad -\infty <x<\infty.
  \label{gfunction}\ee

\begin{prop}   Let $f(x)$ and $g(x)$ be defined  by \eqref{ffunction} and \eqref{gfunction}, and let   $\alpha=\nu+m$. Then
\begin{multline}
K_{\mathrm{II}}^{(0)}(\tau;\xi,\eta)  = \left(
 \frac{\xi}{\eta} \right)^{\kappa/2} \frac{1}{\eta-\xi}
    \Big(\xi \eta f'''(\xi)g'''(\eta)
   +\\
   f''(\xi)\left(g(\eta)-(\alpha-2\kappa-\tau-1)g'(\eta)\right)
 +
 g''(\eta)  \left(f(\xi) + (\alpha-2\kappa-\tau+1)f'(\xi)\right) \\- (\xi+\eta+\alpha\kappa-\kappa^2)f''(\xi)g''(\eta)-f'(\xi)g'(\eta)
\Big),  \label{IIintegrablekernel}
\end{multline}
and moreover, $f(x)$  and $g(x)$ are respectively
  particular solutions of the fourth order ODEs
  \be x^{2} f^{(4)} - (\alpha-2\kappa-1)xf''' -(2x+\alpha \kappa-\kappa^2)f''+(\alpha-2\kappa-\tau+1)f'+f=0 \label{feqn}\ee
and
  \be x^{2} g^{(4)} + (\alpha-2\kappa+1)xg''' -(2x+\alpha \kappa-\kappa^2)g''-(\alpha-2\kappa-\tau-1)g'+g=0.\label{geqn}\ee
  \end{prop}

\begin{proof}
For convenience,  we use the shorthand notation   $J(s)=J_{\alpha}( 2\sqrt{\tau/s})$.
First, we see from the Bessel differential equation
\be z^2 J_{\alpha}''(z)+zJ_{\alpha}'(z)+(z^2 -\alpha^2)J_{\alpha}(z)=0\ee
that
\be \frac{d}{ds}\left(\sqrt{\frac{4\tau}{s}} J_{\alpha}'\Big( \sqrt{\frac{4\tau}{s}}\Big)\right)=\left(\frac{4\tau}{2s^2}-\frac{\alpha^2}{2s}\right) J(s). \label{besselrelation2}\ee
Together with the formula
\be  \sqrt{\frac{4\tau}{s}} J_{\alpha}'\Big( \sqrt{\frac{4\tau}{s}}\Big) =-2s \frac{d}{ds} J(s), \label{besselrelation2}\ee
simple calculations give us
\begin{align}\left(\frac{\partial}{\partial s}+\frac{\partial}{\partial t}\right)&\left(J_{\alpha}\Big( \sqrt{\frac{4\tau}{s}}\Big)\sqrt{\frac{4\tau}{t}} J_{\alpha}'\Big( \sqrt{\frac{4\tau}{t}}\Big)-J_{\alpha}\Big( \sqrt{\frac{4\tau}{t}}\Big)\sqrt{\frac{4\tau}{s}} J_{\alpha}'\Big( \sqrt{\frac{4\tau}{s}}\Big)\right)\nonumber\\
&\ =2(s-t)\frac{\partial}{\partial s}J(s) \frac{\partial}{\partial t}J(t)+ \left(\frac{2\tau}{t^2}-\frac{\alpha^2}{2t}-\frac{2\tau}{s^2}+\frac{\alpha^2}{2s}\right)J(s)J(t). \label{besselrelation3}\end{align}

Noting the simple fact \be (\eta-\xi)e^{\eta s-\xi t}=\left(\frac{\partial}{\partial s}+\frac{\partial}{\partial t}\right)e^{\eta s-\xi t}, \nonumber \ee
combine   \eqref{besselrelation2} and \eqref{besselrelation3}, integrate  by parts    and  we thus get from  \eqref{kernelsub-2}
that
\begin{align}
 &(  \xi/\eta)^{-\kappa/2}  (\eta-\xi)  K_{\mathrm{II}}^{(0)}(\tau;\xi,\eta)  =
 \frac{1}{2\pi i}\int_{0}^{\infty}  dt  \int_{\mathcal{C}_{\mathrm{0}}}  ds  \,   e^{\eta s-\xi t}  \Big(\frac{\partial}{\partial s}+\frac{\partial}{\partial t}\Big) \bigg\{ \Big(  \frac{ s}{ t}\Big)^{\frac{\alpha}{2}-\kappa}  \nonumber \\
 & \ \times
     \frac{e^{\frac{1}{s}-\frac{1}{t}}}{2(s-t)}\Big(J_{\alpha}\Big( \sqrt{\frac{4\tau}{s}}\Big)\sqrt{\frac{4\tau}{t}} J_{\alpha}'\Big( \sqrt{\frac{4\tau}{t}}\Big)-J_{\alpha}\Big( \sqrt{\frac{4\tau}{t}}\Big)\sqrt{\frac{4\tau}{s}} J_{\alpha}'\Big( \sqrt{\frac{4\tau}{s}}\Big)\Big)\bigg\}  \nonumber \\
     &=
     \frac{1}{2\pi i}\int_{0}^{\infty}  dt  \int_{\mathcal{C}_{\mathrm{0}}}  ds  \,   e^{\eta s-\xi t} e^{\frac{1}{s}-\frac{1}{t}} \Big(  \frac{ s}{ t}\Big)^{\frac{\alpha}{2}-\kappa} \bigg\{     \left(\frac{\tau (s+t)}{s^2 t^2}-\frac{\alpha^2}{4st}\right)J(s)J(t) \nonumber \\
 & \  +  \frac{\partial}{\partial s}J(s) \frac{\partial}{\partial t}J(t) +
  \left(\frac{s+t}{s^2 t^2}-\frac{\frac{\alpha}{2}-\kappa}{st}\right)\left(J(t)s\frac{\partial}{\partial s}J(s)-J(s)t\frac{\partial}{\partial t}J(t)\right)
 \bigg\}. \label{IIsum2}
   \end{align}

Now integrate by parts a second time, with \eqref{ffunction} and \eqref{gfunction} in mind, we arrive at
\be  \int_{0}^{\infty}  dt    \,     t^{-(\frac{\alpha}{2}-\kappa)}
 e^{-\xi t-\frac{1}{t}}     \frac{\partial}{\partial t} J(t)= f'(\xi)+(\frac{\alpha}{2}-\kappa)f''(\xi)-\xi f'''(\xi),  
  \label{ffunction1}\nonumber \ee
  \be  \int_{0}^{\infty}  dt    \,     t^{-(\frac{\alpha}{2}-\kappa)}
 e^{-\xi t-\frac{1}{t}}    \frac{1}{t} \frac{\partial}{\partial t} J(t)= -f(\xi)-(\frac{\alpha}{2}-\kappa+1)f'(\xi)+\xi f''(\xi),  
  \label{ffunction2}\nonumber \ee
   and
   \be  \frac{1}{2\pi i}   \int_{\mathcal{C}_{\mathrm{0}}}  ds  \,   s^{\frac{\alpha}{2}-\kappa}
 e^{\eta s+\frac{1}{s}}       \frac{\partial}{\partial s} J(s)= g'(\eta)-(\frac{\alpha}{2}-\kappa)g''(\eta)-\eta g'''(\eta),  
  \label{gfunction1}\nonumber \ee
     \be  \frac{1}{2\pi i}   \int_{\mathcal{C}_{\mathrm{0}}}  ds  \,   s^{\frac{\alpha}{2}-\kappa}
 e^{\eta s+\frac{1}{s}}      \frac{1}{s} \frac{\partial}{\partial s} J(s)= g(\eta)-(\frac{\alpha}{2}-\kappa-1)g'(\eta)-\eta g''(\eta).  
  \label{gfunction2}\nonumber \ee
  Substitution of  the above formulas  into  \eqref{IIsum2}, careful calculations result in the desired formula \eqref{IIintegrablekernel}.

  Next, we turn to the proof of  \eqref{geqn} while that of \eqref{feqn} is similar. Recalling \eqref{besselrelation2}, integrate  by parts two times  and we get
  \begin{align}
 &\frac{1}{2\pi i}   \int_{\mathcal{C}_{\mathrm{0}}}  ds  \,   s^{\frac{\alpha}{2}-\kappa}
 e^{xs+\frac{1}{s}}       \frac{d}{ds}\left(\sqrt{\frac{4\tau}{s}} J_{\alpha}'\Big( \sqrt{\frac{4\tau}{s}}\Big)\right)    \nonumber \\
 & \ =- \frac{1}{2\pi i}   \int_{\mathcal{C}_{\mathrm{0}}}  ds  \,   s^{\frac{\alpha}{2}-\kappa}
 e^{\eta s+\frac{1}{s}}
  \left(x-\frac{1}{s^2}+\frac{\frac{\alpha}{2}-\kappa}{s}\right)(-2s) \frac{d}{ds}J(s)
     \nonumber \\
     &=
     - \frac{2}{2\pi i}   \int_{\mathcal{C}_{\mathrm{0}}}  ds  \,   s^{\frac{\alpha}{2}-\kappa}
 e^{x s+\frac{1}{s}} J(s)\left(x^2 +\frac{1}{s^2} +s
  \Big(x-\frac{1}{s^2}+\frac{\frac{\alpha}{2}-\kappa}{s}\Big)^2\right)
   \nonumber \\
 & \ = - 2  \left(x^2 g^{(4)}+(\alpha-2\kappa+1)xg''' +\big((\frac{\alpha}{2}-\kappa)^2-2x\big) g''-(\alpha-2\kappa-1)g'+g\right). \nonumber 
   \end{align}
      By using \eqref{besselrelation2}, the RHS of the above equation is equal to
     $ 2\tau g'-\frac{\alpha^2}{2}g''.  $  Then \eqref{geqn} follows.
 \end{proof}

Finally, we conclude this subsection with  a remark on the two functions $f(x)$ and $g(x)$  given in \eqref{ffunction} and \eqref{gfunction}.
Define the pairing
\begin{align}
 [f(x), g(y)]   & =      xy f'''(x)g'''(y)
   +
   \big(g(y)-(\alpha-2\kappa-\tau-1)g'(y)-yg''(y)\big)f''(x)
 \nonumber\\
  &\quad  +\big(f(x) + (\alpha-2\kappa-\tau+1)f'(x)-xf''(x)\big)g''(y) \nonumber \\
  &\quad- (\alpha\kappa-\kappa^2)f''(x)g''(y)-f'(x)g'(y)
   \end{align}
 and denote $[f, g](x) = [f(x), g(x)]$ which is the bilinear concomitant. Then

\begin{align}
  \frac{d}{dx}[f, g](x)  &=
   \Big(g(x)-(\alpha-2\kappa-\tau-1)g'(x)-(2x+\alpha \kappa -\kappa^2)g''(x)\nonumber\\
   &\quad +(\alpha-2\kappa+1)xg'''(x)+x^2 g^{(4)}(x)\Big)f'''(x)
 \nonumber\\
  &\quad +
   \Big(f(x)+(\alpha-2\kappa-\tau+1)f'(x)-(2x+\alpha \kappa -\kappa^2)f''(x)\nonumber\\
   &\quad -(\alpha-2\kappa-1)xf'''(x)+x^2 f^{(4)}(x)\Big)g'''(x).
   \end{align}
This shows that  the bilinear concomitant $[f, g](x)$ is constant whenever  f and g satisfy the respective differential equations.

\begin{acknow}   We are   grateful to Gernot Akemann, Peter J. Forrester, Jiang Hu, Mario Kieburg, Dong Wang  and Lun Zhang for helpful discussions. We also thank the anonymous referees for their careful reading and constructive suggestions. 
The work   was  partially supported by the National Natural Science Foundation of China  \#11301499,   the Youth Innovation Promotion Association CAS  \#2017491,  Anhui Provincial Natural Science Foundation \#1708085QA03 and the Fundamental Research Funds for the Central Universities (Grants WK0010450002 and WK3470000008).

\end{acknow}

\bibliographystyle{amsplain}

\end{document}